\documentclass[12pt]{article}
\usepackage{amsthm}
\usepackage{amssymb,amsfonts, bm}
\usepackage{multirow,array}
\usepackage{graphicx,color,rotating}
\usepackage{psfrag}
\usepackage{epstopdf}
\usepackage{subfigure}

\usepackage{amsmath}

\newcommand*\figpath{.}

\def\Bmp#1{ \begin{minipage}{#1} }
\def\Bmpc#1{ \begin{minipage}[c]{#1} }
\def\Bmpt#1{ \begin{minipage}[t]{#1} }
\def\Bmpb#1{ \begin{minipage}[b]{#1} }
\def\Emp{ \end{minipage} }

\newtheorem{theorem}{\bf Theorem}[section]

\usepackage{amsfonts}
\usepackage{graphicx}
\usepackage{subfigure}
\usepackage{epstopdf}
\usepackage{algorithmic}
\ifpdf
  \DeclareGraphicsExtensions{.eps,.pdf,.png,.jpg}
\else
  \DeclareGraphicsExtensions{.eps}
\fi

\newcommand{\gap}{\mbox{\hspace{1em}}}
\newcommand{\dee}[1]{\,{\rm d}{#1}}

\newcommand{\ord}{\mathcal{O}}

\newcommand{\J}{\mathcal{J}}
\def\RR{{\mathbb{R}}}
\newcommand{\argmin}{\operatorname{argmin}}
\newcommand{\supp}{\operatorname{supp}}

\makeatletter

\newcommand{\Rmnum}[1]{\expandafter\@slowromancap\romannumeral #1@}
\makeatother
\usepackage{authblk}

\title{On the convergence of data assimilation for the one-dimensional
  shallow water equations with sparse observations}

\author{N.~K.-R.~Kevlahan\footnote{Corresponding author. Email
      address: kevlahan@mcmaster.ca}, R.~Khan, B.~Protas}
  \affil{Department of Mathematics and Statistics, McMaster University
    \\ Hamilton, Ontario, L8S 4K1, Canada} \date{}

\begin{document}
\maketitle


\pagestyle{myheadings}
\thispagestyle{plain}
\begin{abstract}
  The shallow water equations (SWE) are a widely used model for the
  propagation of surface waves on the oceans.  We consider the problem
  of optimally determining the initial conditions for the
  one-dimensional SWE in an unbounded domain from a small set of
  observations of the sea surface height.  In the linear case we prove
  a theorem that gives sufficient conditions for convergence to the
  true initial conditions. At least two observation points must be
  used and at least one pair of observation points must be spaced more
  closely than half the effective minimum wavelength of the energy
  spectrum of the initial conditions. {This result also applies to the
    linear wave equation.}  Our analysis is confirmed by numerical
  experiments for both the linear and nonlinear SWE data assimilation
  problems.  These results show that convergence rates improve with
  increasing numbers of observation points and that at least three
  observation points are required for the practically useful results.
  Better results are obtained for the nonlinear equations provided
  more than two observation points are used. This paper is a first
  step in understanding the conditions for observability of the SWE
  for small numbers of observation points in more physically realistic
  settings.
\end{abstract}


\section{Introduction\label{sec:intro}}
This paper considers computational aspects of a variational data
assimilation technique that could be applied to improve forecasting of
tsunami waves. More specifically, we focus on how to choose the number
and locations of wave height sensors so that iterative gradient-based
approaches typically employed to solve data assimilation problem
converge to correct solutions. We derive an observability-type
criterion based on a linearized partial differential equation model
and show how it can be used to characterize the accuracy of
reconstructions in both the linear and nonlinear setting.  Although we
consider the idealized one-dimensional case, the Nyquist-like
observability condition we derive for the spacing of the observation
points should carry over to the more realistic two-dimensional case.
	
Current data assimilation methods for tsunamis usually aim to
reconstruct the initial sea surface height perturbation from seismic
data.  This requires using seismic data to find an optimal
representation for the motion of the seafloor and then using a
separate technique, such as the Okada model~\cite{Okada:1985}, to
convert this seismic data into an initial condition for the sea
surface displacement and velocity.  These seismic models are often
supplemented with real-time sea surface data from tsunameters to make
them more robust.  For example, the MOST model~\cite{Titov/etal:2005}
uses an assimilation method developed at the Pacific Marine
Environmental Laboratory (PMEL) that includes both seismic data and
direct observations of the sea surface perturbations.  The PMEL model
does not use variational assimilation, but instead finds the best
linear combination from a database of 246 pre-computed unit source
propagation solutions given the seismic data and tsunameter
observations.  This provides estimates of the tsunami wave
field in deep water than can be used to initialize inundation models
for coastal areas.

Maeda et al.~\cite{Maeda/etal:2015} propose an alternative
approach that assimilates data from a network of sea surface height
buoys into a 2D linear SWE model to estimate the full
tsunami wave field at a given point in time. Although this
approach uses sea surface observations, it is in fact a 3D-VAR
method in the sense that it only considers spatial dependence of the
data. Thus, since there is no explicit time dependence, this
approach does not really reconstruct the initial conditions.

We use the linear and nonlinear shallow water equations (SWE) as the
mathematical model for wave propagation, although our approach could
be extended easily to other one- or two-dimensional models (e.g.\
Boussinesq). The SWE are a coupled system of equations for
non-dispersive travelling waves where the wavelength $\lambda$ is much
larger than the ocean depth $h$, allowing averaging over the depth.

Previous work on variational data assimilation schemes for the SWE
include Stefanescu et al.~\cite{Stefanescu/etal:2015}, who compare
three reduced-order approaches to initial condition estimation,
projecting the dynamical system onto subspaces consisting of basis
elements representing the characteristics of the expected solution and
subsequently using these low-order models as surrogates for the
original system in the optimization process. While this is a novel
approach for the given problem, their fundamental objective is
improving computational efficiency. In contrast, our focus is on the
structure of the observation operator and its effect on the
convergence of the gradient-based solution of the underlying
optimization problem.  Our aim is to find the minimum observational
information required to correctly reconstruct the initial conditions,
whereas they do not address the choice of observation operator in
their analysis.

Similarly, Lacasta et al.~\cite{Lacasta/etal:2017} also focus on
computational issues, addressing the memory requirements of solving an
adjoint system where information about the physical system is stored
at all times. They compare a brute-force Monte Carlo method with
gradient-based optimization utilizing the adjoint equations. Results
are gauged by the number of simulations required to converge to the
desired state.  Analysis of optimal locations and/or the numbers of
sensors to record water depth is not performed; indeed, they 
use a single observation point.
    
In summary, any overlap between this work and the afore-mentioned studies
is limited to the derivations of the adjoint system and
the gradient formulation, which are well known. We emphasize that the
observability-type criterion derived in this study presents a
significant addition to the previous works which aimed to quantify and  improve computational efficiency.

For simplicity,  we  focus on the one-dimensional (1D) problem
defined on the real line, so that the space variable $x \in \RR$.  The
prognostic variables are height $h(x, t) = H + \eta(x,t)$, where $H$
is the average depth and $\eta(x,t)$ is the perturbation of free
surface, and velocity $u(x,t)$.   We 
consider the shallow water system
\begin{subequations}
\label{eq:SWEnl}
\begin{align}
\eta_t + (h u)_x = 0, \label{eq:SWEnla} \\
 u_ t +(\tfrac{1}{2} u^2 + \eta)_x = 0, \label{eq:SWEnlb} \\
 \eta(x,0) = \phi(x) &, \gap u(x, 0) = 0, \label{eq:SWEnlc} 
\end{align}
\end{subequations}
and its linearized version 
\begin{subequations}
\label{eq:SWEl}
\begin{align} 
\eta_t + u_x = \ 0 & , \label{eq:SWEla} \\
u_t + \eta_x = \ 0 &,  \label{eq:SWE1b} \\
\eta(x,0) = \phi(x) &, \gap u(x, 0) = 0 \label{eq:SWElc}
\end{align}
\end{subequations}
corresponding to the assumption $|\eta| \ll H$ which is a good
approximation in the case of the deep ocean.  The linear equation is
normalized so that that the mean depth $H=1$ and the wave speed $c=1$.
In both cases we assume that the support of the initial condition
$\phi$ is $\ord(1)$ and that both velocity and the perturbed surface
height vanish at infinity resulting in the boundary conditions
\begin{equation}
\eta(x,t), \ u(x,t) \rightarrow 0 \quad  \text{as} \quad x \rightarrow \pm\infty, \ t>0.
\label{eq:BC}
\end{equation}
We also assume constant depth (no bathymetry). 

In this paper we are interested in the following fundamental question:
what is the minimum information, in the form of observations of the
surface height perturbation $\eta$, required to correctly reconstruct
the {true} initial condition $\phi^{(t)}$ using a variational data
assimilation approach?

We assume that time-resolved observations {$y^{(o)}_j(t)$,}
$j=1,\dots,N_{\text{obs}}$, of the system evolution described by
\eqref{eq:SWEnl} or \eqref{eq:SWEl} are available at $N_{\text{obs}}
\ge 1$ distinct fixed observation points $\{ x_j
\}_{j=1}^{N_{\text{obs}}}$.  Since we are interested in addressing
certain basic observability questions (rather than problems of
ill-posedness or poor conditioning~\cite{ehn96}), we assume that there
is no noise in the observations, i.e.\,
{
\begin{equation}
y^{(o)}_j(t) = \eta^{(t)}(x_j,t), \quad j=1,\dots,N_{\text{obs}},
\label{eq:m}
\end{equation}
where $\{y^{(o)}_j(t)\}$ are observations of the surface displacement
based on the true solution $\eta^{(t)}(x,t)$ {generated by} the
true initial condition $\phi^{(t)}(x)$.}

This approach, often referred to as ``4D-VAR'', has received much
attention in the literature~\cite{t05}, particularly for applications
to operational weather forecasting \cite{k03} and climate modelling.
Given the ubiquity of the SWE as a model in Earth sciences,
variational data assimilation approaches for this system have already
received significant attention~\cite{znz94,pm01a,Belanger2003,vkn14,TIRUPATHI2016198}.

Indeed, relatively simple area-limited finite difference SWE models
are often used to test algorithms and answer basic questions about
4D-VAR~\cite{Steward/etal:2012}. A key difference between the
4D-VAR formulations used in weather and climate models and
the data assimilation employed for the reconstruction of tsunami waves
motivating the present study is that in the latter case the
observation data is extremely sparse, which immediately raise the question of observability.

In the context of the 4D-VAR approaches applied to the 2D SWE
the question of observability was first addressed by Zou et al~\cite{Zou/Navon/LeDimet:1992}. They investigated this problem in the finite-dimensional (discrete) setting by studying the positive-definiteness of the Hessian of the error functional and by invoking the ``rank'' condition known from the linear control theory~\cite{s94}. Their main finding was that the discretized 2D SWE is
observable even if only one or two of the three variables $(u, v, \eta)$
are observed at {\it all\/} grid points.  However, the authors did not
determine lower bounds for the number of observations (e.g.\ the smallest number and locations of grid points on which one state variable needs to be measured in order to ensure observability).
This result is therefore less relevant for the case of tsunami prediction
where observations are typically available at only a handful of grid
points. Our study resolves this question for the 1D SWE by
determining {\it minimum\/} sufficient conditions on the observations which ensure correct reconstruction the initial conditions.

It is also important to note that, unlike previous work on 4D-VAR
which is applied to discretized (finite-dimensional) models, we
consider the continuous infinite-dimensional case.  {Note that
  infinite dimensional data assimilation, based on the evolution of a
  partial differential equation, are often referred to as distributed
  parameter systems.}  As a result, our findings are independent of
discretization and are therefore more general. It should also be
emphasized that the techniques used by Zou et
al~\cite{Zou/Navon/LeDimet:1992} to study observability do not
generalize in a straightforward manner to infinite
dimensions~\cite{z95}. We therefore have to rely on other approaches.

We add that here we are not interested primarily in the optimal choice
of observations, or which gradient minimization algorithm should be
used for optimal convergence (although we did verify that our
results are qualitatively the same for both gradient descent and
conjugate gradient algorithms).  However, our findings do
provide fundamental insights about the structure of the observational
data needed to make 4D-VAR approaches operationally viable in cases,
like tsunami prediction, where only a small number of sparsely
distributed observation locations is available.

There has already been significant progress made in analyzing the
observability and controllability of one-dimensional waves in
{\it bounded} domains.  Unlike parabolic problems, in which
information propagates at an infinite velocity, such questions are
more subtle in the case of the wave equation because of the finite
speed at which information travels in such systems~\cite{Zuazua:2005}.
As a result, the observation time $T$ tends to play an important role.

One such case is the ``boundary measurements'' problem where observations
of wave height are acquired at one end of a finite domain.  In
general, exact observability is achieved for linear, semilinear and
quasilinear wave equations provided the observation time $T\geq2L/c$,
where $L$ is the size of the spatial domain and $c$ is the wave
speed~\cite{Zuazua:2005,Li:2006,Guo:2010}; this property can also be
deduced from analogous controllability results~\cite{Coron2007}.
Intuitively, observability (i.e. recovering the true initial
conditions from observations) is possible under this criterion because
all information from both right-going and left-going waves is
recovered during the observations, as the waves have enough time to
reflect from the boundaries.  

In contrast, here we consider multiple observation locations on an {\it infinite\/} domain, but on one
side only of the support of the initial condition.  Since in the
absence of the boundaries the waves do not reflect, in this case it is
not clear {\it a priori\/} whether waves moving in opposite directions
are observable, even with multiple observation points. Nevertheless,
we demonstrate that in fact observability is achieved with two or more
observation locations provided the locations are chosen appropriately
(i.e.\ at least one pair of observation points must be sufficiently
close together). Numerical computations are then used to illustrate
how this criterion affects the accuracy of the variational data
assimilation both in the linear and nonlinear setting.

The structure of the paper is as follows: in the next section we
recall the standard gradient-based approach to solution of the
variational data assimilation problem for systems~\eqref{eq:SWEnl} or
\eqref{eq:SWEl}, where we follow the ``optimize-then-discretize''
paradigm~\cite{g03}; then, in Section~\ref{sec:analytic}, an analytic
solution of the data-assimilation problem in the linear setting is
presented which leads, in Section~\ref{sec:fixedpoint}, to a criterion
ensuring convergence of iterations to the true solution. Computational
tests validating and illustrating this criterion are then presented in
Section~\ref{sec:num}, whereas conclusions and final comments are
deferred to Section~\ref{sec:concl}.

\section{Data assimilation for the one-dimensional shallow water equations\label{sec:assimil}}
Our goal is to use a variational optimization approach to construct
the best approximation {$\phi^{(b)}$} to the true initial
condition {$\phi^{(t)}$} in \eqref{eq:SWEnlc} or
\eqref{eq:SWElc} given observations $\{y^{(o)}_j(t)\}$,
cf.~\eqref{eq:m}. It is assumed that all observation points are
located on one side (e.g. to the right) of the support of the true
initial data, i.e.,
\begin{equation}
\forall j=1,\dots, N_{\text{obs}} \qquad x_j > \operatorname{max}_x \left[ \supp \phi^{(t)}(x) \right],
\label{eq:xj}
\end{equation}
so that only right-going waves can pass through the observation
points, and the observations are made over the interval $0<t\leq T$,
where $T$ is the observation time. {The notation used for the
  data assimilation is summarized in Table~\ref{tab:notation}.}
\begin{table}
\begin{tabular}{|l|l|}
\hline
{\bf Symbol} & {\bf Definition} \\
\hline
$\eta(x,t)$ & general solution for the height perturbation \\
$\phi(x)$ & general initial condition, i.e.\ $\phi(x) := \eta(x,0)$\\
$\eta^{(t)}(x,t)$ & true solution for the height perturbation $\eta(x,t)$ \\
$\phi^{(t)}(x)$ & true initial condition, i.e.\ $\phi^{(t)}(x) := \eta^{(t)}(x,0)$ \\
$\phi^{(g)}(x)$ & starting guess for the initial condition \\
$\phi^{(n)}(x)$ & approximate initial condition at iteration $n$ of the assimilation algorithm, i.e.\ $\phi^{(0)}:=\phi^{(g)}$\\
$\phi^{(b)}(x)$ & best approximation to the initial condition (e.g.\ fixed point of iterations)\\
$y^{(o)}_j(t)$ & observations of the true height perturbation at positions $\{x_j\}, j = 1,\ldots,N_{\text{obs}}$ \\
$\eta^{(f)}(x,t)$ & approximate (``forecast'') solution generated by an approximate initial condition  \\
${\cal J}^{(n)}$ & cost function at iteration $n$\\
$(\cdot)^*$ & adjoint \\
\hline
\end{tabular}
\caption{\label{tab:notation} {Notation used in data assimilation of the SWE to find the optimal initial conditions from observations
    of the height perturbation. We have followed the notation proposed in \cite{Ide/etal:1997}.}}
\end{table}

As is typical in data assimilation problems, the question of finding
{the best estimate of} the initial data {$\phi^{(b)}$}
such that the corresponding solution of system \eqref{eq:SWEnl} or
\eqref{eq:SWEl} optimally matches the available observations
$y^{(o)}_j(t)$, $j=1,\dots, N_{\text{obs}}$, is framed as a
PDE-constrained optimization problem where a least-squares error
functional $\J \; : \; L^2(\RR) \rightarrow \RR$ defined as
\begin{equation}
\J(\phi) = \frac {1}{2} \int_0^T \sum_{j=1}^{N_\text{obs} } \left[\eta^{(f)}(x_j,t) - y^{(o)}_j(t) \right]^2 \, \dee{t}
\label{eq:cost}
\end{equation}
is minimized with respect to the initial data $\phi$ under the
constraint that the `{`forecast"} solution $\eta^{(f)}(\cdot,\cdot)$ {corresponding to the initial condition $\phi$}
satisfies system \eqref{eq:SWEnl} or \eqref{eq:SWEl}. Then, since $\eta^{(f)}$ is the 
solution of \eqref{eq:SWEnl} or \eqref{eq:SWEl} corresponding to $\phi$, we obtain the following reduced
(unconstrained) formulation
\begin{equation}
\phi^{(b)} = \argmin_{\phi \in L^2(\RR)} \J(\phi).
\label{eq:minJ}
\end{equation}

We remark that in actual applications the error functional
\eqref{eq:cost} is typically augmented with a suitable Tikhonov
regularization term which provides stability in the presence of
observation noise \cite{ehn96} {and makes the problem well-posed.} {We also add that t}{T}he assumption that
the initial data $\phi$ belongs to the space $L^2(\RR)$ of
square-integrable functions on $\RR$ is {enough} {sufficient} to ensure existence of
(possibly non-smooth) solutions to the initial-value problems
\eqref{eq:SWEnl} and \eqref{eq:SWEl}. Local minimizers $\phi^{(b)}$
are then characterized by the condition expressing the vanishing of the
G\^ateaux (directional) differential of $\J(\phi)$, i.e,
\begin{equation}
\forall _{\phi' \in L^2(\RR)} \quad \J'(\phi^{(b)};\phi') = 0,
\label{eq:dJ}
\end{equation}
where $\J'(\phi;\phi') = {\lim_{\epsilon\rightarrow 0}}\left[\J(\phi+\epsilon\phi')
  -\J(\phi)\right]/\epsilon$. We add that in general condition \eqref{eq:dJ}
describes only critical points of the error functional. The local
minimizers can be found as $\phi^{(b)} = \lim_{n\rightarrow
  \infty}\phi^{(n)}$ using an iterative gradient-descent procedure
\begin{equation}
\begin{aligned}
\phi^{(n+1)} &= \phi^{(n)} - \gamma_n \nabla\J(\phi^{(n)}), \\
\phi^{{(0)}} &= \phi^{(g)},
\end{aligned}
\label{eq:iter}
\end{equation} 
where $\phi^{(g)}$ is a suitable initial guess. A key element of the
iterative procedure \eqref{eq:iter} is evaluation at every iteration
$n$ of the gradient $\nabla\J(\phi^{(n)})$ which represents an
infinite-dimensional sensitivity of the error functional
\eqref{eq:cost} with respect to perturbations of the initial data
$\phi$. It can be obtained from the G\^ateaux differential
$\J'(\phi;\phi')$ by invoking the Riesz representation theorem
as~\cite{b77}
\begin{equation}
\forall _{\phi' \in L^2(\RR)} \quad \J'(\phi;\phi') = 
\big\langle \nabla\J(\phi), \phi' \big\rangle_{L^2(\RR)},
\label{eq:Riesz}
\end{equation}
where $\langle f, g \rangle_{L^2(\RR)} = \int_\RR fg \, dx$ is the
standard $L^2(\RR)$ inner product. A convenient expression for the
gradient is then obtained using the adjoint calculus as
\begin{equation}
\nabla \J = - \eta^* (x,0),
\label{eq:gradJ}
\end{equation}
where $\eta^*$ satisfies 
\begin{subequations}
\label{eq:aSWEl}
\begin{align} 
\eta^*_t + u^*_x =& -  \sum_{j=1}^{N_\text{obs} }\left[\eta^{(f)}(x_j,t) - y^{(o)}_j(t) \right] \delta(x-x_j), \label{eq:aSWEla} \\
u^*_t + \eta^*_x =& 0, \label{eq:aSWElb} \\
\eta^*(x,T) = \ {u^*}(x,T) =& 0, \label{eq:aSWElc}
\end{align}
\end{subequations}
when the system evolution is governed by the {\it linear} shallow
water system \eqref{eq:SWEl}, and
\begin{subequations}
\label{eq:aSWEnl}
\begin{align} 
\eta^*_t + u \eta^*_ x  + u^*_x =& -  \sum_{j=1}^{N_\text{obs} }\left[\eta^{(f)}(x_j,t) - y^{(o)}_j(t) \right] \delta(x-x_j), \label{eq:aSWEnla} \\
u^*_t + (1+\eta)\eta^*_x + u u^*_x =& 0, \label{eq:aSWEnlb} \\
\eta^*(x,T) = \ {u^*}(x,T) =& 0, \label{eq:aSWEnlc}
\end{align}
\end{subequations}
when the system evolution is governed by the {\it nonlinear\/} shallow
water system~\eqref{eq:SWEl}. In both cases the adjoint variables
$\eta^*$ and $u^*$ vanish at infinity sufficiently rapidly.

Since the derivations leading to relations
\eqref{eq:gradJ}--\eqref{eq:aSWEnl} are standard and involve only
integration by parts together with making suitable choices of the
source terms and initial/boundary data in
systems~\eqref{eq:aSWEl}--\eqref{eq:aSWEnl}, we skip these steps here
for brevity and refer the reader to the monograph~\cite{g03} for
details.  We add that when the initial data $\phi$ belongs to a
(Hilbert) space of functions with higher regularity, such as one of
the Sobolev spaces $H^p(\RR)$, $p \ge 1$, then a similar formalism
based on the Riesz representation~\eqref{eq:Riesz} can be followed to
determine the corresponding Sobolev gradients~\cite{pbh04}.

In the gradient-descent formula~\eqref{eq:iter} the step size $\gamma_n$ along the gradient direction 
$\nabla\J(\phi^{(n)})$ is computed optimally by solving, at every iteration $n$, a line-minimization problem
\begin{equation}
\gamma_n = \argmin_{\gamma > 0} \J( \phi^{(n)} - \gamma \nabla\J(\phi^{(n)}) )
\label{eq:gamman}
\end{equation}
which is done efficiently using standard line-search
techniques~\cite{nw00}. Finally, we add that, while in actual
applications one typically uses more efficient minimization approaches
such as the conjugate-gradient or quasi-Newton methods, here we focus
on the simple gradient descent approach~\eqref{eq:iter} as it provides
a clear perspective on the observability issues which are the main
focus of this study.

In the next section we show how solutions to the minimization
problem~\eqref{eq:minJ} can be characterized analytically in function
of the available observation data $y^{(o)}_j(t)$, $j=1,\dots,
N_{\text{obs}}$, when the system evolution is governed by the linear
SWE \eqref{eq:SWEl}. This allows us to establish conditions under
which iterations~\eqref{eq:iter} converge to the true initial data
{$\phi^{(t)}$}.

\section{Analytic solution of the data assimilation problem in the linear setting\label{sec:analytic}}
In Section~\ref{sec:assimil} we obtained expression~\eqref{eq:gradJ}
for the gradient of the error functional and the adjoint
system~\eqref{eq:aSWEl} corresponding to the linear
SWE~\eqref{eq:SWEl}.  We now solve these equations analytically to
find the exact form of the gradient of the error functional under the
assumption of linear evolution.  The adjoint system~\eqref{eq:aSWEl}
must be solved backwards in time from $t=T$ to $t=0$ to find the
gradient which is defined in terms of its solution evaluated at $t=0$,
cf.~\eqref{eq:gradJ}.

For convenience, we introduce the backwards time variable $\tau = T-t$
and define vectors 
{
\begin{eqnarray}
U & = & [\eta^*\gap u^*]^T, \\
F & = & \left[-\sum_{j=1}^{N_{\text{obs}}}\left[\eta^{(f)}(x_j,t) - y^{(o)}_j(t) \right] \delta(x-x_j) \gap 0 \right]^T, \label{eq:force}
\end{eqnarray}
where $F$ is the source term for the adjoint system.} In
terms of these new variables the adjoint equations
\eqref{eq:aSWEla}--\eqref{eq:aSWElb} become
\begin{equation}
U_\tau + A U_x = -F, \label{eq:adj2}
\end{equation}
where
\[
A = \left[\begin{array}{rr} 0 & -1\\-1 & 0 \end{array}\right].
\]
We now diagonalize \eqref{eq:adj2}, obtaining
\begin{equation}
V_\tau + D V_x = d, \label{eq:adj3}
\end{equation}
where
\[
V=P^{-1}U, \quad d = -P^{-1} F, \quad
P = \left[\begin{array}{rr} 1 & -1\\1 & 1 \end{array}\right], \quad \text{and} \quad  
D = \left[\begin{array}{rr} -1 & 0\\0 & 1 \end{array}\right].
\]
The two characteristics of the hyperbolic system \eqref{eq:adj3} arriving at point $(x,\tau)$ are
\begin{align}
\begin{split}
X_1(\tau,x,T) & =  -\tau + x +T,\\
X_2(\tau,x,T) & =  \tau+x -T.
\end{split}
\label{eq:characteristics}
\end{align}
Since the initial (or terminal in terms of the original variable $t$) condition for the adjoint system \eqref{eq:adj3} is zero, its
solution may be written as~\cite{John:1982}
\begin{equation}
V_i(x,T) = 0 + \int_0^T d_i(X_i(\tau,x,T),\tau)\, \dee{\tau}, \gap i=1,2.
\label{eq:back_sol}
\end{equation}
The gradient of the error functional, cf.~\eqref{eq:gradJ}, can now be expressed {in terms of the solution (\ref{eq:back_sol})} as 
{
\begin{eqnarray}
\nabla \J & = & -\eta^*(x,\tau=T) \nonumber \\
& = & -U_1(x,T)\nonumber\\
& = &  -[PV]_1 = V_2(x,T) - V_1(x,T), 
\label{eq:gradJ2}
\end{eqnarray}
}
where the subscripts~$[\cdot]_1$ {and $[\cdot]_2$} denote components of the relevant vector.   

{Therefore, to find the gradient of the error functional at the
  first iteration we simply need to evaluate the integrals
  (\ref{eq:back_sol}).  This requires expressing $d_i$ in terms of the
  source term (\ref{eq:force}) as $d = -P^{-1}F$ and then integrating
  from $\tau = 0$ to $T$ along the characteristics
  (\ref{eq:characteristics}), i.e.\ back to the initial state $t = 0$.
  Using {the symbolic algebra software} {\tt maple} to evaluate
  this integral, }the gradient at the first iteration is {found
  to be}
\begin{equation}
\nabla \J(x;\phi^{(0)}) = -\frac{N_{\text{obs}}}{4} \left[ (\phi^{(t)}(x) - \phi^{(0)}(x))  + 
\frac{1}{N_{\text{obs}}}\sum_{j=1}^{N_{\text{obs}}} (\phi^{(t)}(2x_j-x) - \phi^{(0)}(2x_j-x))\right]. 
\label{eq:analytic_grad}
\end{equation}
in which the dependence of the gradient on the independent variable $x$ had to be made explicit. {(Note that the gradient may also be derived in a less straightforward way from \eqref{eq:SWEl} using the solution of the linear wave equation and Duhamel's principle for inhomogeneous partial differential equations.)} Then, the gradient descent iteration~\eqref{eq:iter} for the data-assimilation problem in the linear setting is
\begin{multline}\label{eq:grad_descent}
\phi^{(n+1)}(x) = \phi^{(n)}(x) +  \\
\gamma_n\frac{N_{\text{obs}}}{4}\left[ (\phi^{(t)}(x) - \phi^{(n)}(x))  + \frac{1}{N_{\text{obs}}}
\sum_{j=1}^{N_{\text{obs}}} (\phi^{(t)}(2x_j-x) - \phi^{(n)}(2x_j-x))\right]. 
\end{multline}
In the next Section we analyze the conditions under which the iterative sequence $\{\phi^{(n)}\}$, $n=1,2,\dots$,  converges to the true initial condition {$\phi^{(t)}$}.

\section{Sufficient conditions for convergence to true initial condition\label{sec:fixedpoint}}
Our goal in this section is to establish conditions under which the
iterates {$\phi^{(n)}$} generated by the gradient descent
algorithm~\eqref{eq:grad_descent} converge to the true solution {$\phi^{(t)}$}.
We assume that the Fourier transform {$\widehat{\phi^{(t)}}(k)$} of
the true initial data, where $k \in \RR$ is the wavenumber, is
compactly supported, i.e.,
\begin{equation}
\exists \ k_{\text{max}}>0 \; :  \quad {\widehat{\phi^{(t)}}(k) }= 0 \ \text{if} \ |k| > k_{\text{max}}.
\label{eq:fhat}
\end{equation}
The main result can be stated as the following theorem.
\begin{theorem}
\label{theorem:fixedpoint}
{Suppose Then, t}{T}he iterative sequence~\eqref{eq:grad_descent} converges in the $L^2(\RR)$ norm to the true initial data {$\phi^{(t)}$}, i.e.,
\begin{equation}
\lim_{n \rightarrow \infty} \| \phi^{(n)} - {\phi^{(t)}} \|_{L^2(\RR)} = 0
\label{eq:convergence}
\end{equation}
provided the following condition holds
\begin{equation}
|\widehat{\psi}(k)| = \left|\frac{1}{N_{\text{obs}}} \sum_{j=1}^{N_{\text{obs}}} e^{i k 2 x_j}\right|  \neq1
\quad \forall k \in \RR. \label{eq:fp_condition}
\end{equation}
\label{thm:conv}
\end{theorem}
\begin{proof}
Equation~(\ref{eq:grad_descent}) shows that the fixed point  \[\phi^{(b)}= \lim_{n \rightarrow \infty} \phi^{(n+1)} =  \lim_{n \rightarrow \infty} \phi^{(n)}\] satisfies the relation
\begin{equation}
(\phi^{(t)}(x) - \phi(x))  + \frac{1}{N_{\text{obs}}}\sum_{j=1}^{N_{\text{obs}}} (\phi^{(t)}(2x_j-x) - \phi(2x_j-x)) = 0. \label{eq:fp1}
\end{equation}
Defining the reconstruction error at the fixed point as 
$g = \phi^{(t)}-\phi^{(b)}$ and taking the Fourier transform of
\eqref{eq:fp1} we obtain
\begin{equation}
\widehat{g}(k) + \frac{1}{N_{\text{obs}}} \sum_{j=1}^{N_{\text{obs}}} e^{i k 2 x_j} \widehat{g}^*(k) = 0  
\gap  \forall k \in \RR,\label{eq:fp2}
\end{equation}
where $\widehat{g}(k)$ is the Fourier transform of $g(x)$, $i=\sqrt{-1}$
and $^*$ now indicates complex conjugate. Denoting $\widehat{\psi}(k) :=
\frac{1}{N_{\text{obs}}}\sum_{j=1}^{N_{\text{obs}}} e^{i k 2 x_j}$ and
taking the complex conjugate of~(\ref{eq:fp2}), we find that
$\widehat{g}^*(k) = \widehat{\psi}^*(k) \, \widehat{g}(k)$ and so the Fourier
transform of the reconstruction error at the fixed point must satisfy
the relation
\begin{equation}
\widehat{g}(k) = |\widehat{\psi}(k)|^2 \, \widehat{g}(k) \gap \forall  k \in \RR.\label{eq:fp3}
\end{equation}
Equation~(\ref{eq:fp3}) implies that {\it either\/}
$|\widehat{\psi}(k)|=1$ {\it or\/} $\widehat{g}(k) = 0$, in which case the
fixed point $\phi^{(b)}$ is precisely the true initial condition
{$\phi^{(t)}$}.  Therefore, relation~\eqref{eq:fp_condition} represents a
sufficient condition for the existence of a unique fixed point
corresponding to the true initial condition. Convergence in the
$L^2(\RR)$ topology follows from the isometry property of the Fourier
transform as a map $L^2(\RR) \rightarrow L^2(\RR)$.
{\qed}
\end{proof}

We now comment on the the insights provided by Theorem~\ref{thm:conv}.
First, we note that if $|\widehat{\psi}(k)|=1$ then any solution
$\widehat{g}(k)$ is possible, resulting in arbitrary forms of the
reconstruction error. In other words, if $|\widehat{\psi}(k)|=1$ for any
wavenumber $k$, the actual fixed point (i.e.\ the {best estimate
  for the} reconstructed initial condition $\phi^{(b)}$) depends
on the initial guess {$\phi^{(g)}$} and the gradient descent
iteration~\eqref{eq:grad_descent} will not, in general, converge to
the true initial condition {$\phi^{(t)}$}. We add that
expression~\eqref{eq:fp_condition} could be simplified by assuming
that the observation points $x_j$ are equispaced.

We now analyze the conditions under which
relation~(\ref{eq:fp_condition}) might not be satisfied, i.e.\ the
conditions where the gradient descent
algorithm~(\ref{eq:grad_descent}) might not converge to the correct
initial data. {We emphasize that for iterations
  \eqref{eq:grad_descent} to converge to the correct initial data, the
  condition $|\widehat{\psi}(k)| \neq 1$ must hold for {\it any}
  wavenumber $k$ and therefore, in principle, it needs to be analyzed
  for all values of this parameter. However, due to
  assumption~\eqref{eq:fhat} and the form of the gradient
  expression~\eqref{eq:analytic_grad}, which involves a linear
  combination of the true initial condition together with its shifted
  and scaled copies, the Fourier transform of the reconstruction
  error, $\widehat{g}(k)$, vanishes for $k>k_{\text{max}}$. Thus, we
  can restrict the range of wavenumbers for which condition
  \eqref{eq:fp_condition} needs to be tested to $0 \le k \le
  k_{\text{max}}$.}  

First consider the case of a single observation point,
$N_{\text{obs}}=1$.  In this case the condition for a non-unique fixed
point is
\begin{equation}
|e^{i k 2 x_1} |= 1. \label{eq:not_unique1}
\end{equation}
Clearly, condition~(\ref{eq:not_unique1}) is always satisfied
{regardless of the value of $k$} and the gradient descent
iteration~\eqref{eq:grad_descent} never converges to the true
initial condition (unless the initial guess {$\phi^{(g)}$} is exact).
Intuitively, there is not enough information in a single observation
to distinguish between the right-going and left-going parts of the SWE
solution.  There is an infinite set of initial conditions $\phi$
compatible with the observations and the actual fixed point found
depends on the initial guess {$\phi^{(g)}$}. For example, if
{$\phi^{(g)}=0$}, the fixed point solution is {$\phi^{(b)}(x) =
\frac{1}{2}(\phi^{(t)}(x) + \phi^{(t)}(2x_1-x))$}.

Next, consider the case of two observation points, $N_{\text{obs}}=2$.
In this case the condition for a non-unique fixed point is
\begin{equation}
\left|\frac{1}{2} e^{ik2x_1}\left(1+e^{ik2(x_2-x_1))} \right)\right| = 1. \label{eq:not_unique2}
\end{equation}
To satisfy condition (\ref{eq:not_unique2}) we require that
$e^{ik2(x_2-x_1)}=1$ or, equivalently,
\begin{equation}
x_2-x_1 = \frac{\pi}{k} n, \gap n = {0,1,2,}\ldots, \gap \forall k \in \RR. \label{eq:not_unique3}
\end{equation}
Therefore, a sufficient condition for the existence of a unique fixed
point corresponding to the true initial condition is
that~{(\ref{eq:not_unique3})} be not be satisfied for any
wavenumber $k$ for which $\widehat{g}(k)$ is non-zero. {Due to
  assumption~\eqref{eq:fhat} and the form of the gradient
  expression~\eqref{eq:analytic_grad}, which involves a linear
  combination of the true initial condition together with its shifted
  and scaled copies, the Fourier transform of the reconstruction
  error, $\widehat{g}(k)$, also vanishes for $k>k_{\text{max}}$.
  Therefore, we can guarantee that~(\ref{eq:not_unique3})} {We
  see that it} is never satisfied provided the spacing of the
observation points satisfies
\begin{equation}
x_2-x_1 < \frac{\pi}{k_{\text{max}}}, \label{eq:unique}
\end{equation}
or, $x_2-x_1 < \lambda_{\text{min}}/2$ where $\lambda_{\text{min}}$ is
the smallest ``scale'' of the true initial condition
{$\phi^{(t)}$}. Intuitively, this means that a sufficient
condition for convergence to the true initial condition is that the
observation points are closer than half the minimum wavelength of the
initial condition, which can be interpreted as analogous to the
Nyquist criterion.  Finer scale initial conditions mean the
observation points should be closer together.

The case of more than two observation points, $N_{\text{obs}}\geq3$,
is similar to the case of two observation points.  However, for more
than two observation points a sufficient condition for the existence
of a unique fixed point is that at least one pair of observation
points satisfies relation~(\ref{eq:unique}). This is indeed sufficient
to ensure condition~(\ref{eq:fp_condition}) is always satisfied since
\[
\left|\frac{1}{N_{\text{obs}}}\sum_{j=1}^{N_{\text{obs}}} e^{i k 2 x_j}\right| = 1
\] 
only if {\it every\/} pair of observation points
$(x_i, x_j)$ gives $|e^{i k 2(x_i-x_j)} |=1$.  We conclude that, in
general, if there exists a pair of observation points violating
condition~\eqref{eq:unique}, then the Fourier components of the
reconstructed initial data $\widehat{\phi}^{(b)}$ corresponding to the
wavenumbers $k_n = n \pi / (x_2 - x_1)$, $n=1,2,\dots$, will not be
reconstructed correctly.

From the computational point of view, we expect that if
condition~\eqref{eq:fp_condition} is satisfied, but the quantity
$|\widehat{\psi}(k)|$ is close to one for some $0 < k < k_{\text{max}}$,
the convergence of the corresponding Fourier components of the
reconstructed initial data $\widehat{\psi}$ will be slow. {This is
  because the quantity $|\psi(k)|$ may be interpreted as a
  wavenumber-dependent ``constant'' characterizing the linear rate of
  converge of the gradient iterations \eqref{eq:grad_descent} for
  individual Fourier components of the solution.}  
  
To close
this section, in Figure~\ref{fig:psihat} we {explore} the
behaviour of $|\widehat{\psi}(k)|$ {(and hence the convergence rate
  of the data assimilation)} {for equispaced observation points
  as a function of different problem parameters, namely, the
  wavenumber $k$, the number of observation points $N_{\text{obs}}$ 
  and their separation $\Delta x$. Since we
  assumed that $k_{\text{max}} = 30$, we consider a range of discrete
  wavenumbers $k \in \{ \Delta k, 2\Delta k, \dots, k_{\text{max}} \}$
  for two different resolutions $\Delta k = 1$ and $\Delta k =0.075$
  (which correspond to discretizations with grid sizes $h=\pi/\Delta
  k$ in physical space).  Therefore, the separation $\Delta x$ of 
  adjacent observation points must be less than $\pi/ k_{\text{max}} =
  \pi/30 \approx 0.1$ to recover the initial condition. In
  Figure~\ref{fig:psihat} we see that, indeed, the condition
  $|\widehat{\psi}(k)| < 1$ is satisfied for all $0 < \Delta x < \pi /
  k_{\text{max}}$ regardless of the values of $N_{\text{obs}}$ and
  $k$, which confirms the results of the analysis above.}

{Figure~\ref{fig:psihat} provides two key insights
  complementing the theoretical analysis. First, for
increasing $N_{\text{obs}}$, $|\widehat{\psi}(k)|$ decreases for all
$0 < \Delta x \lessapprox 0.1$. This suggests that increasing
  the number of observation points improves the rate of
  convergence.  Secondly, the results for different $k$ show
  that even when $0 < \Delta x < \pi / k_{\text{max}}$ (i.e.\ when the
  initial condition can be recovered exactly), $|\widehat{\psi}(k)|$ is
  arbitrarily close to one when the separation of the points $\Delta
  x$ is sufficiently small.  This suggests that observation points
  with vanishingly small spacing decelerate convergence and
    should be avoided.  
    
    Interestingly, in the discrete
  case, there are many observation point spacings larger than the
  critical spacing $\Delta x < \pi / k_{\text{max}}$ that allow 
   convergence to the correct initial data. This
  is because the ``bad'' spacings form only a sparse subset of
  the interval $[\pi / k_{\text{max}}, \infty )$, although the subset becomes
  increasingly space-filling as the resolution is refined, i.e., in
  the limit $\Delta k \rightarrow 0$.}
\begin{figure}
\begin{center}
\includegraphics[width=0.5\textwidth]{\figpath/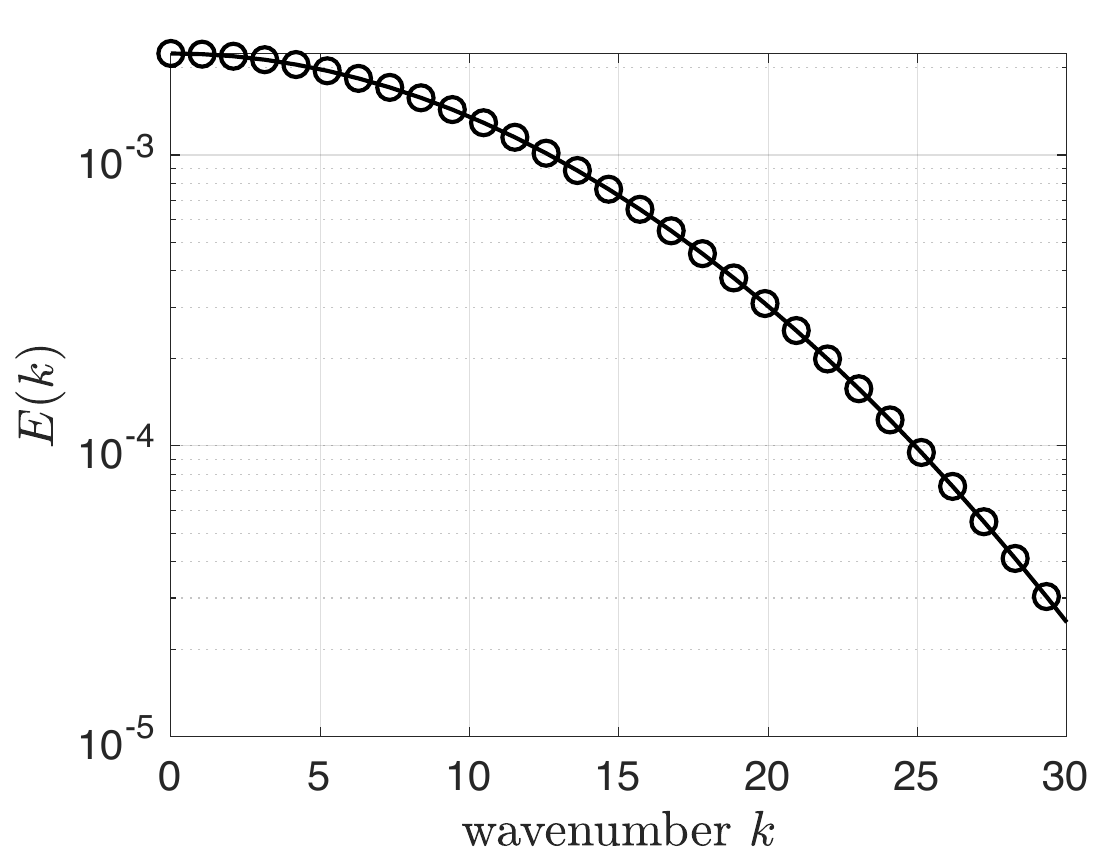}
\end{center}
\caption{Fourier energy spectrum of the height perturbation initial
  condition {$\phi^{(t)}(x) = \exp(-(10x)^2)/20$}. For {wavenumbers}
  $k>30$ the spectral energy densities are less than 1\% of their
  maximum values which suggests taking $k_{max}=30$ is a reasonable
  approximation of the support of this particular initial
  condition in the Fourier space, cf.~\eqref{eq:fhat}. This means that
  {observation points with separations} {$\Delta x <
    \pi/30 \lessapprox 0.1$} should satisfy the sufficient condition
  for convergence given in \eqref{eq:fp_condition}.}
\label{fig:ghat}
\end{figure}
\begin{figure}
\subfigure[$N_{\text{obs}} = 2$, $\Delta k =1.$]{\includegraphics[width=0.45\textwidth]{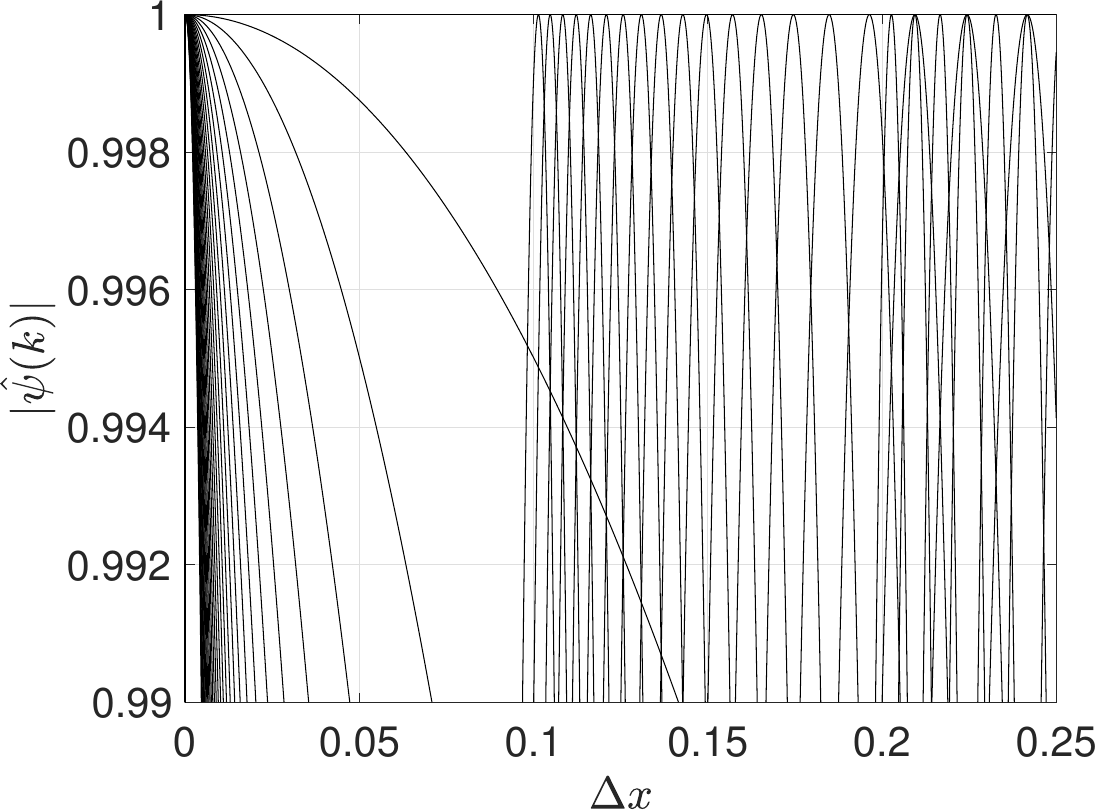}}\quad
\subfigure[$N_{\text{obs}} = 2$, $\Delta k =0.075.$]{\includegraphics[width=0.48\textwidth]{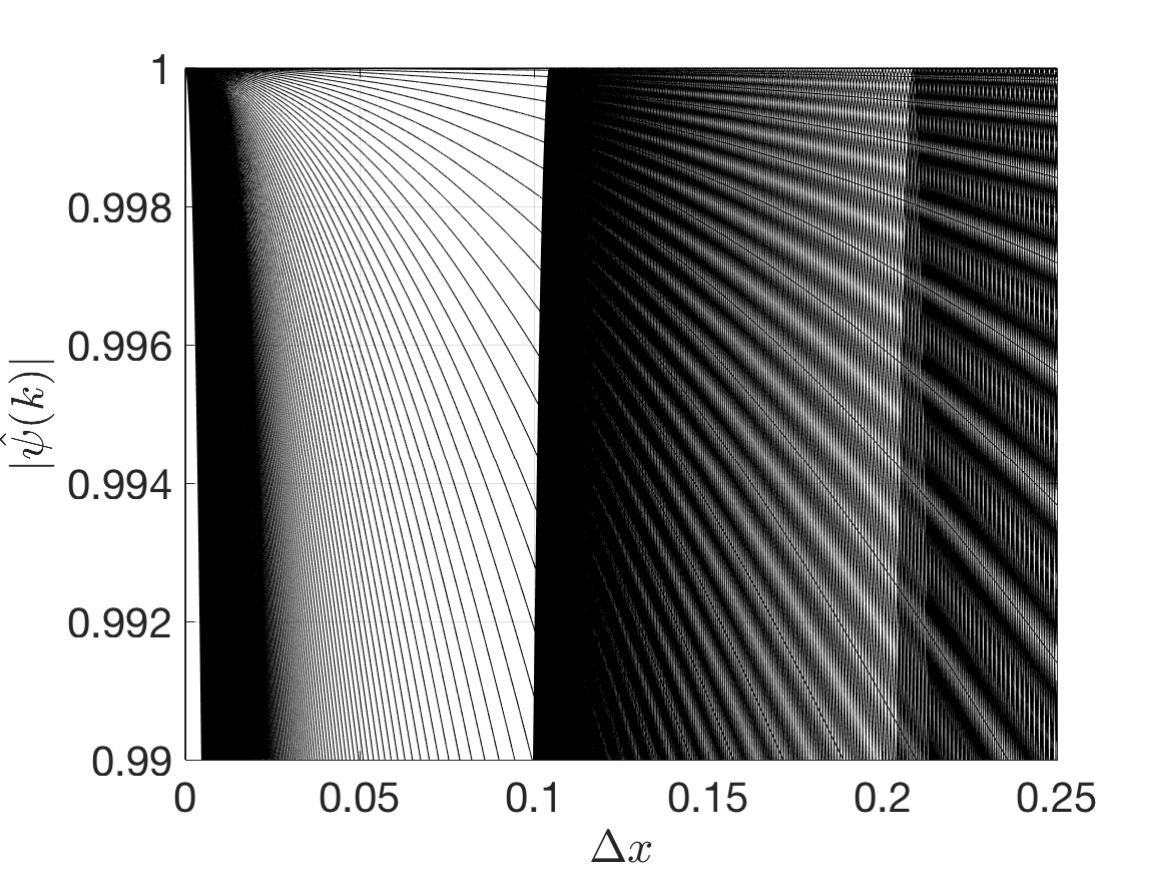}}
\subfigure[$N_{\text{obs}} = 4$, $\Delta k =1.$]{\includegraphics[width=0.45\textwidth]{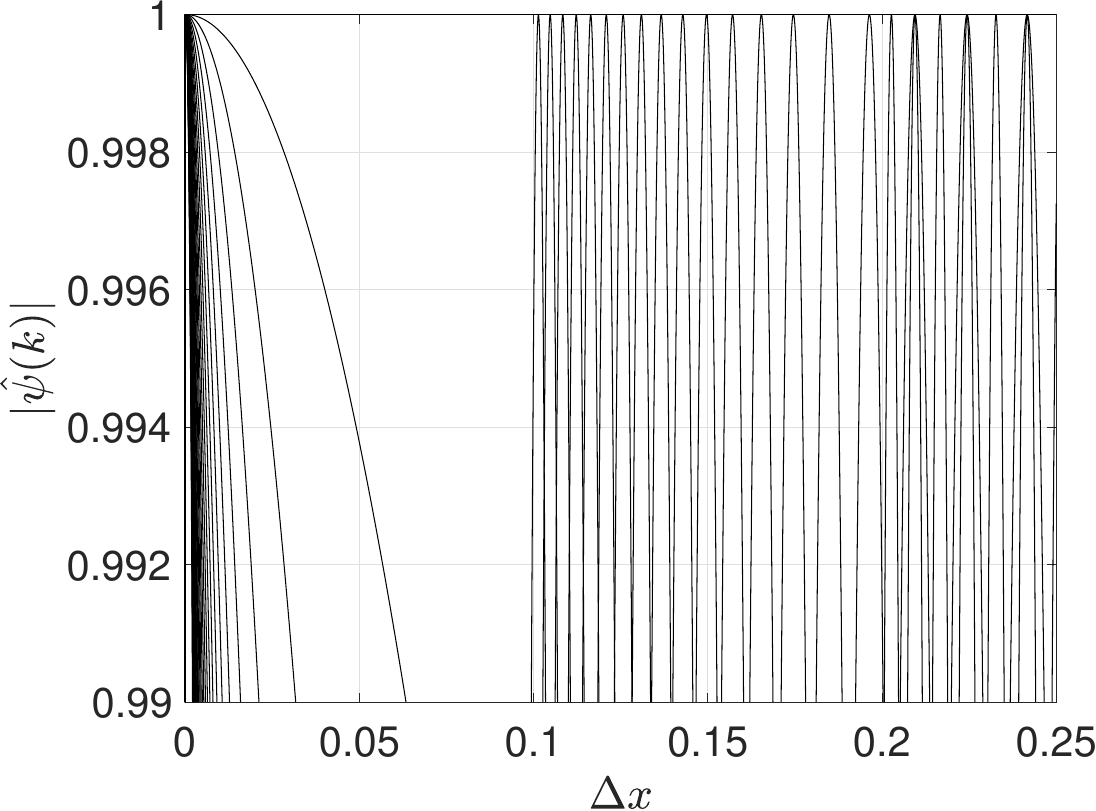}}\quad
\subfigure[$N_{\text{obs}} = 4$, $\Delta k =0.075.$]{\includegraphics[width=0.48\textwidth]{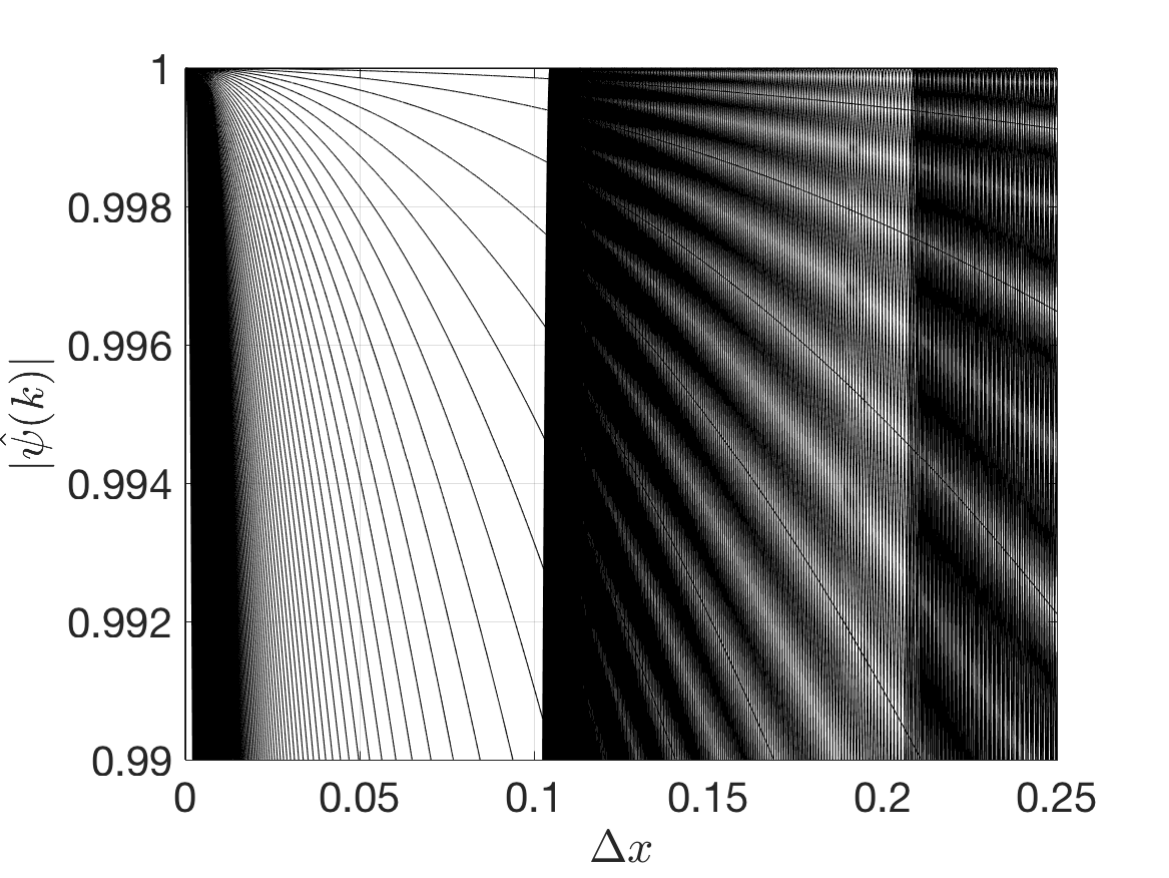}}
\subfigure[$N_{\text{obs}} = 100$, $\Delta k =1.$]{\includegraphics[width=0.45\textwidth]{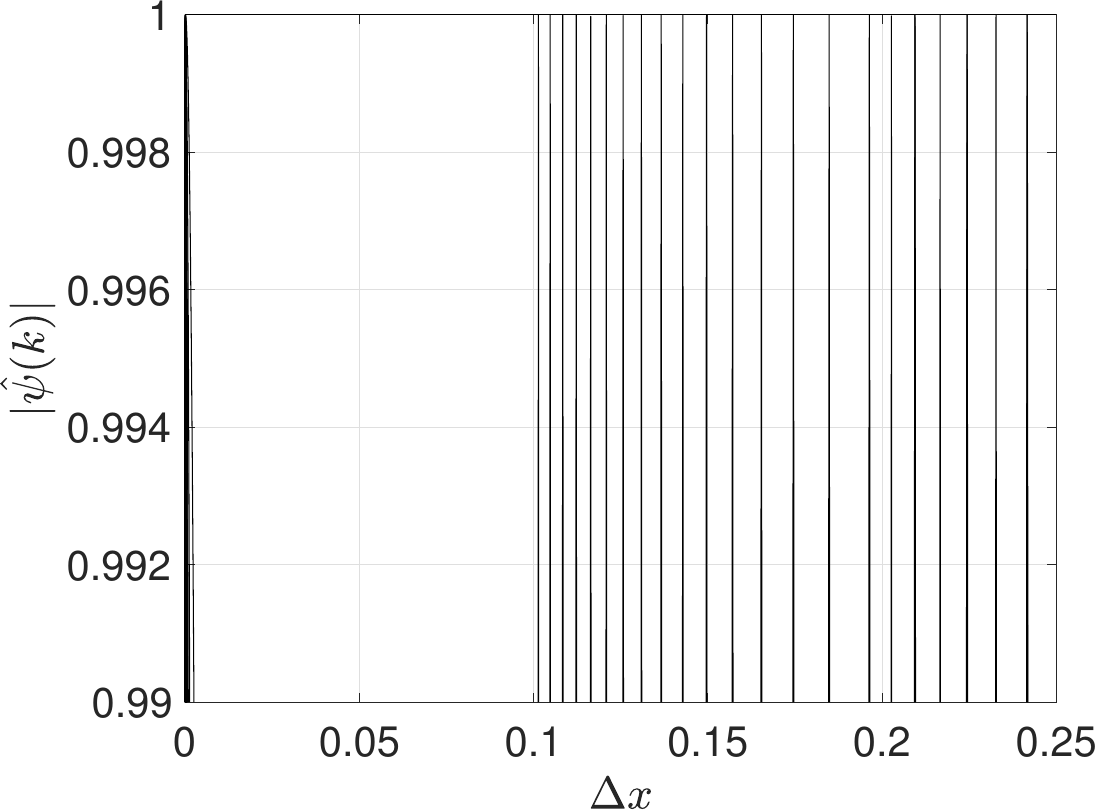}}\quad
\subfigure[$N_{\text{obs}} = 100$, $\Delta k =0.075.$]{\includegraphics[width=0.45\textwidth]{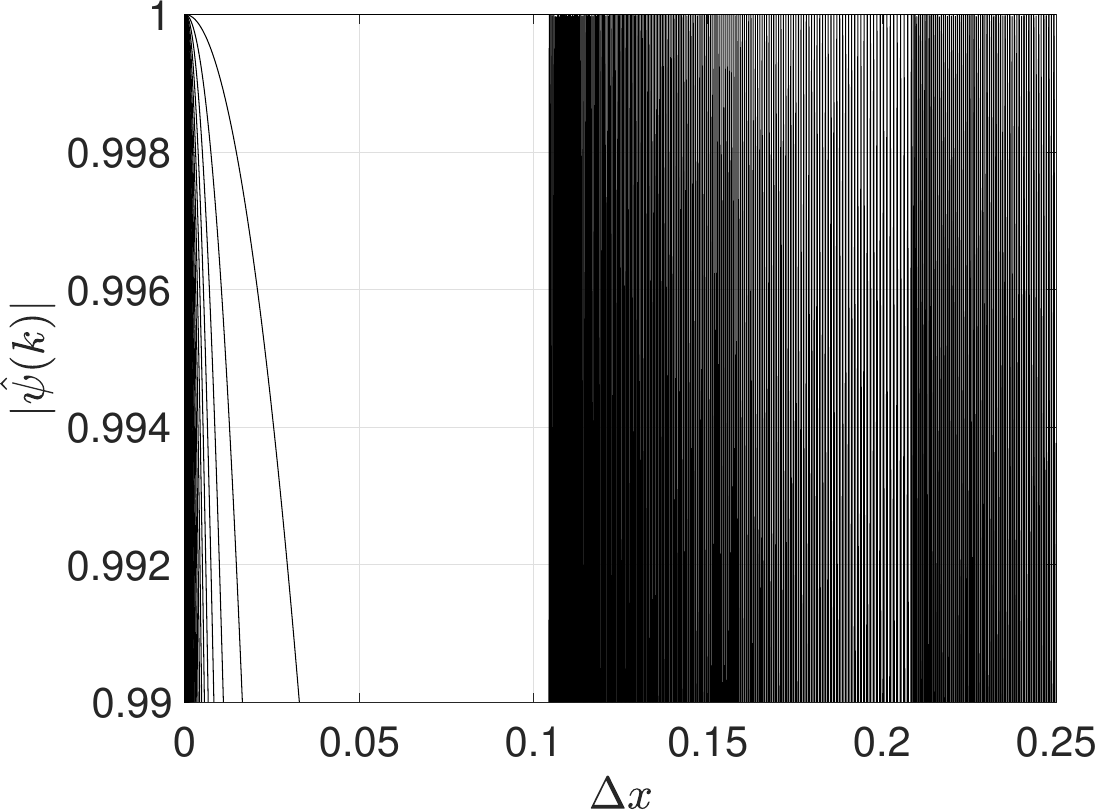}}
\caption{{Convergence properties of the {gradient approach
      \eqref{eq:grad_descent}} for varying numbers of observation
    points $N_{\text{obs}}$ when the real line is discretized in
    wavenumber space with resolution $\Delta k$.  Where
    $|\widehat{\psi}(k)|=1$, cf.~\eqref{eq:fp_condition}, the assimilation
    does not converge to the correct initial data and when
    $|\widehat{\psi}(k)|$ is close to 1 we expect convergence to be slow.
    A clear transition between convergence to the correct and
      incorrect initial data is evident at $\Delta x = \pi /
      k_{\text{max}}$ (where $k_{\text{max}} = 30$ fully resolves the
      initial condtion, Figure \ref{fig:ghat}).  Note that in this
    case the continuous problem converges for all $0<\Delta x<0.1$. In
    all cases the maximum wavenumber $k_{\text{max}} = 30$ to fully
    resolve the exact initial condition (see figure~\ref{fig:ghat}).
    We see that, when $\Delta x < \pi / k_{\text{max}}$, the
      convergence increases with increasing number of observation
    points and also with coarser discretizations of the real line.}}
\label{fig:psihat}
\end{figure}

In the next section we present computational examples illustrating how
these properties affect the accuracy and efficiency of the
gradient-based solution of the data assimilation problem for the wave
equation both in the linear and nonlinear regime.

\section{Numerical solution of the data assimilation problem\label{sec:num}}
Before presenting our computational results we briefly describe the
numerical approach. While in the linear setting both the direct and
adjoint problem \eqref{eq:SWEl} and \eqref{eq:aSWEl} can be solved
analytically, cf.~Section \ref{sec:analytic}, for ease of comparison
with the nonlinear case these problems are solved numerically.

\subsection{Numerical Methods\label{sec:methods}}
In order to simplify the numerical solution of the PDE problems, the
unbounded spatial domain is replaced with a large periodic domain
$[-L,L]$, where $L=3$ is chosen big enough to ensure that in
combination with the observation time $T=2$ there are no boundary
effects (i.e. the waves never get close to the domain boundary during
the time window $[0,T]$). Both the forward and backward (adjoint)
equations are discretized in space using a standard second-order
finite-difference/finite-volume method on an evenly spaced staggered
grid with $N=1024$ points. The height perturbation and velocity
variables $\eta(x,t)$ and $u(x,t)$ are represented at the cell centres
and at the cell boundaries, respectively.  The resulting system of
ordinary differential equations is integrated in time using a
third-order strong stability preserving Runge--Kutta
method~\cite{Spiteria/Ruuth:2002}. Consistency of adjoint-based
gradient evaluations in the linear and nonlinear setting was carefully
checked by evaluating the directional derivative $\J'(\phi;\phi')$
using formulas \eqref{eq:Riesz}--\eqref{eq:gradJ} and comparing the
results with a finite-difference approximation of
$\J'(\phi;\phi')$~\cite{Khan_MScThesis}. In the gradient-based
optimization algorithm \eqref{eq:iter} the line-minimization problem
\eqref{eq:gamman} is solved using the MATLAB function {\tt fminunc}.

\subsection{Computational Results\label{sec:results}}
In this section we present computational results illustrating the
effect of the observability criterion introduced in
Section~\ref{sec:fixedpoint} on the accuracy and efficiency of the
gradient-based solution of the data assimilation problem. 

In all cases
the true initial condition for the height variable $\eta(x,0)$ has the
form 
\[{\phi^{(t)}(x)} = \frac{1}{20}\exp^{-(10x)^2}\]
and is approximately supported on the
interval $[-0.3,0.3]$ in physical space and the interval $[-30,30]$ in
Fourier space, see figure~\ref{fig:ghat} (the initial condition for
the velocity variable $u(x,0)$ is zero, cf.~\eqref{eq:SWEnlc} and
\eqref{eq:SWElc}). We note that the choice of the support in Fourier space coincides with the value of $k_{\text{max}}$ used in Section~\ref{sec:fixedpoint}. Such localized initial conditions are
typical for the problem of tsunami waves propagation.  We emphasize
that the support of the initial condition {$\phi^{(t)}(x)$} in physical space is
ten times smaller than the size of the computational domain.

The true initial condition {$\phi^{(t)}$} is used in conjunction with the forward
model~\eqref{eq:SWEnl} or \eqref{eq:SWEl} and the observation operator
defined in~\eqref{eq:m} to generate the observations $y^{(o)}_j(t)$,
$j=1,\dots,N_{\text{obs}}$ used in the data assimilation problem
(note that in our computations the observations are generated using the same model as the forward equations; SWE in both cases). In all cases we take a zero initial guess $\phi^{(0)}(x) = 0$ to initialize
the gradient iterations~\eqref{eq:iter}.  Figure~\ref{fig:solutions}
shows the right-going parts of the linear and nonlinear SWE
solutions at the observation time $T=2$.  Both solutions are
effectively zero at the boundary of the truncated computational
  domain ($L=3$) at this time.  Recall that for the linear SWE the
solution is simply a translated copy of the initial condition scaled
by one half.
\begin{figure}
\centering
\includegraphics[width= 0.6\textwidth]{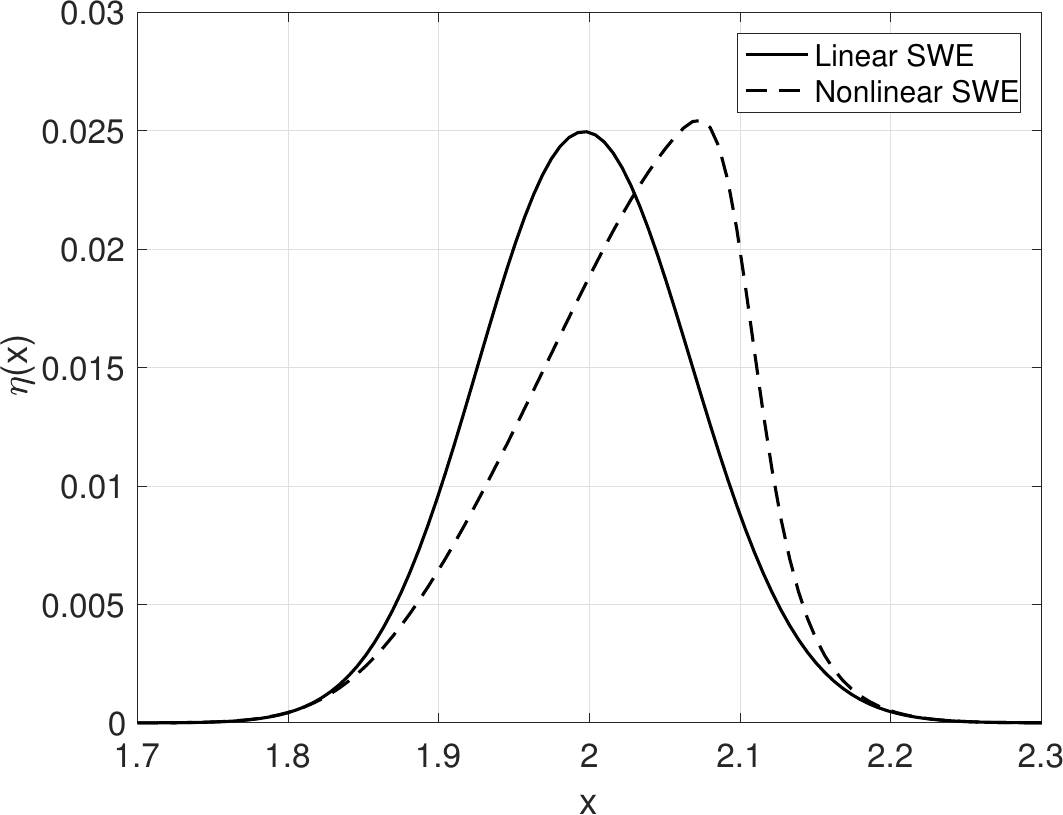}
\caption{Right-going part of the linear and nonlinear SWE solutions at
  the observation time $T=2$. Note the significant difference
  between the two solutions.\label{fig:solutions}}
\end{figure}
The large computational domain and relatively short
observation time ensure that the artificial (periodic)
boundary conditions do not affect the solution (i.e. the domain
remains effectively infinite).  Therefore, choosing all
observation points in the region $x>0$ means that they measure only
the right-going part of the SWE wave solution and do not capture any
information about the left-going part of the solution, as is
typical for the case of tsunami observations.

Next, we present results for the convergence of the cost
functional and of the $L^2$ error in the reconstructed
initial conditions for both the linear and nonlinear SWE cases for
different numbers of observations.  It is important to remember that
we minimize the cost functional~(\eqref{eq:cost}), which is the
$L^2$ norm of the difference between the observations and model
integrated over the observation time, rather than the initial
condition itself.  We stress that, as demonstrated in Section~\ref{sec:fixedpoint}, the assimilation algorithm will not
necessarily converge for all choices of observation points.

Taking Theorem~\ref{thm:conv} and Figure~\ref{fig:ghat} as guides, we
compare a case where the sufficient condition for convergence is
satisfied ($\Delta x = {0.09}$) with a case where it is not
($\Delta x = 0.375$).  In both cases the algorithm is stopped after
1000 iterations. Figure~\ref{fig:Jerr_good} confirms that when the
sufficient condition is satisfied, the algorithm converges for
different numbers of observation points $N_{\text{obs}}>1$.  A small
value of the cost functional indeed corresponds to a small error in
the reconstructed initial condition. The rate of convergence increases
with the number of observation points, consistent with the observation
that $|\widehat{\psi}(k)|$ is further from one as the number of
observation points increases (see Figure~\ref{fig:psihat}(a, c, e)).
Interestingly, taking only two observation points produces much worse
results than three or more observation points.  This suggests that in
practical applications at least three observation points should be
taken.  Note that the linear and nonlinear SWE data assimilation
problems exhibit similar behaviour, although the error for the
nonlinear case is lower than for the linear case for four or more
observation points.  {We have checked that these results are essentially identical for the case of
an isolated bathymetry feature,  $\beta(x) = \exp(-10(x-3/2)^2)/10$ where the total mean depth is $1-\beta(x)$ (results not shown).}
\begin{figure}
\subfigure[]{\includegraphics[width=0.48\textwidth]{\figpath/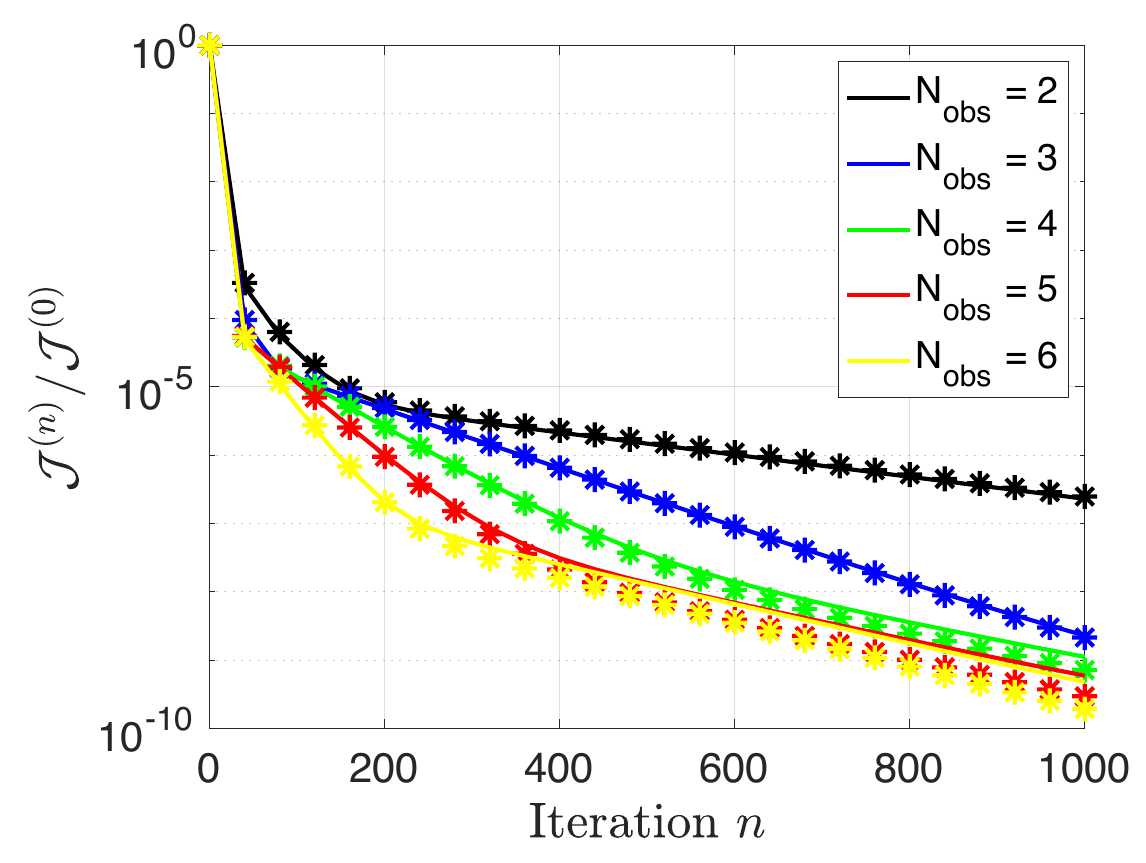}}\quad
\subfigure[]{\includegraphics[width=0.48\textwidth]{\figpath/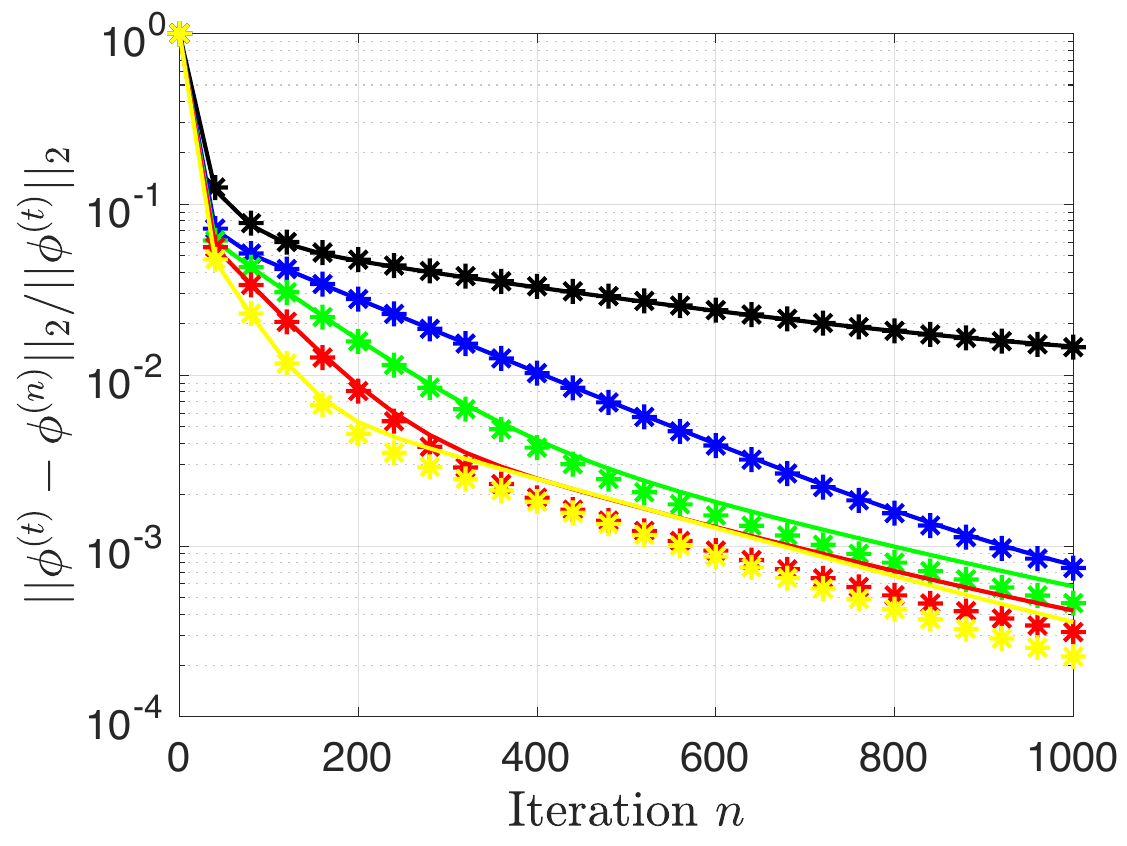}}
\caption{Convergence of (a)~cost function $\J^{(n)}$ and (b) the
  reconstruction error $\|\phi^{(t)} - \phi^{(n)}\|_{L^2(\RR)}$ as functions
  of iteration $n$ for different numbers of observation points for the
  linear {(lines)} and nonlinear {(symbols)}  SWE when the first
  observation point is at $x=0.2$ and the observation points have
  uniform spacing $\Delta x =0.09$. Note that both the cost function
  and reconstruction error converge for all $N_{obs}$, although the
  rate of convergence increases with more observation points. {The results for 
  the linear and nonlinear SWE are  similar.}}
\label{fig:Jerr_good}
\end{figure}

In contrast, when the sufficient condition for convergence
\eqref{eq:fp_condition} is not satisfied, the algorithm fails
to converge to the correct initial condition, or converges very
slowly.  Figure~\ref{fig:Jerr_bad} shows convergence of the cost
functional and error in the reconstructed initial condition
with uniform spacing $\Delta x = 0.375$.  In this case, the error in
the reconstructed initial condition is not reduced significantly for
any number of observation points.  The cost functional is
reduced by five to seven orders of magnitude, although convergence is
extremely slow after the first iteration.  In contrast to the previous
case, the cost functional is much larger for the nonlinear SWE
assimilation, except for the case of two observation points.

\begin{figure}
\subfigure[]{\includegraphics[width=0.48\textwidth]{\figpath/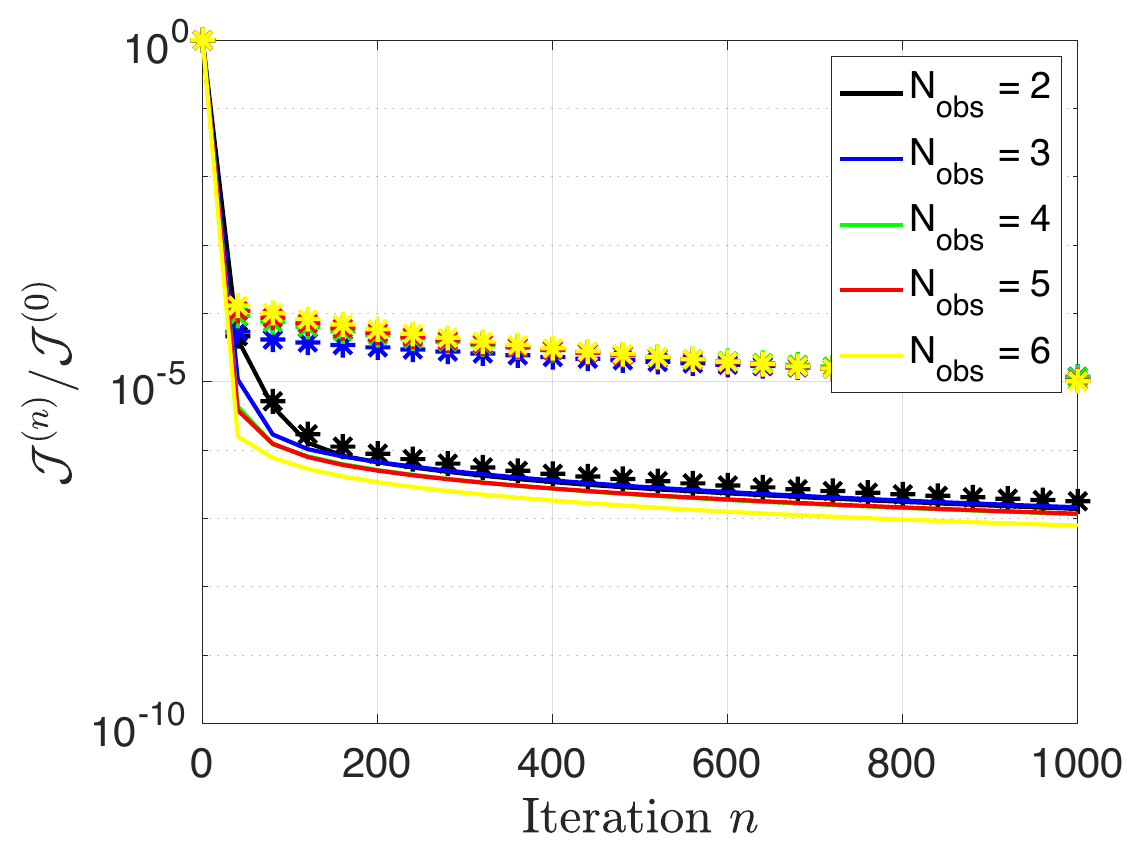}}\quad
\subfigure[]{\includegraphics[width=0.48\textwidth]{\figpath/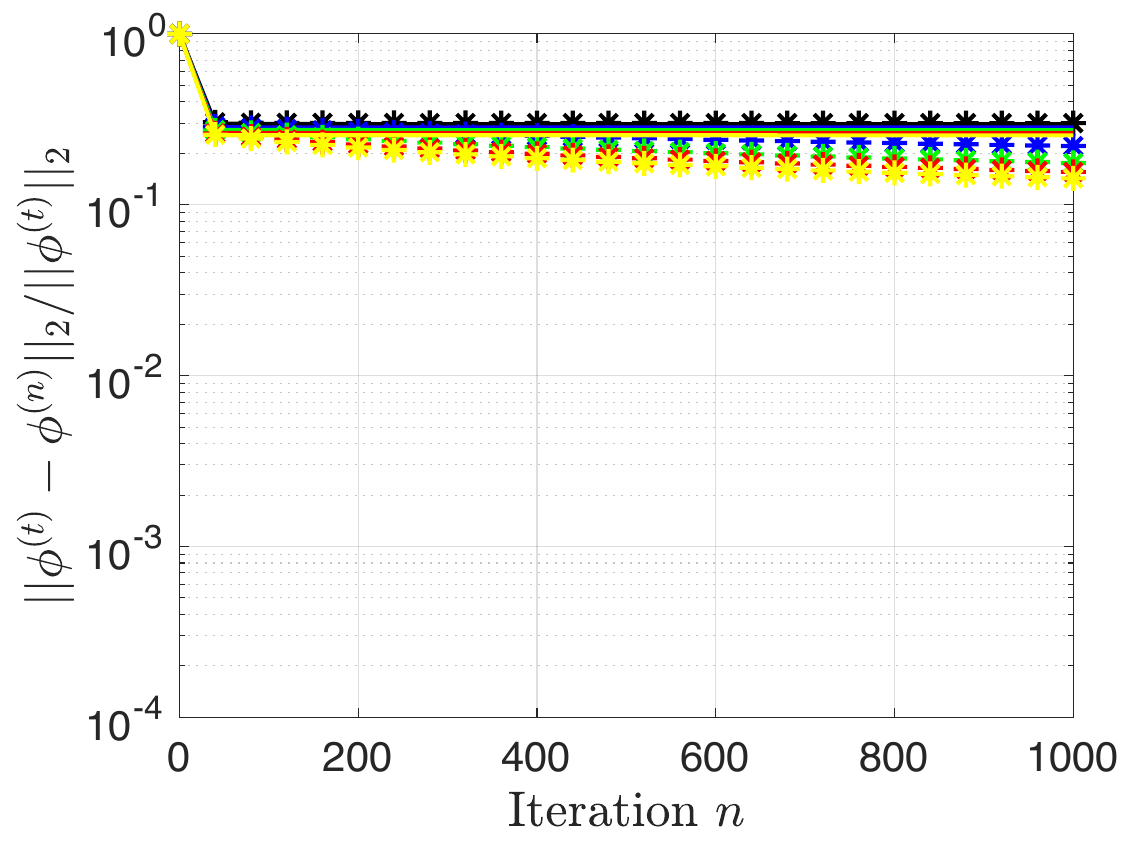}}
\caption{Convergence of (a)~cost function $\J^{(n)}$ and (b)~the
  reconstruction error $\|\phi^{(t)} - \phi^{(n)}\|_{L^2(\RR)}$ as
  functions of iteration $n$ for different numbers of observation
  points for the linear {(lines)} and nonlinear {(symbols)}
  SWE when the first observation point is at $x=0.2$ and the
  observation points have regular spacing $\Delta x =0.375$.  In this
  case $\Delta x$ does not satisfy the sufficient condition for
  convergence \eqref{eq:fp_condition} and the reconstructed initial
  condition fails to converge to the correct initial condition $\phi^{(t)}$.
  {Note that the different data sets overlap since there
    is little dependence on the number of observations points in this
    case where the initial conditions cannot be recovered.}}
\label{fig:Jerr_bad}
\end{figure}

Figure~\ref{fig:optimized} compares the {true initial condition $\phi^{(t)}$} and reconstructed initial conditions after $n=1000$ iterations
{$\phi^{(n)}\approx\phi^{(b)}$} for different spacings of observation
points. In the first case, shown in Figure~\ref{fig:optimized}~(a), the sufficient condition for
convergence to the true initial condition is satisfied
and the initial condition is reconstructed exactly. In the
second case, shown in Figure~\ref{fig:optimized}~(b), the sufficient
condition is not satisfied and the reconstructed initial
condition involves spurious oscillations spread throughout the
computational domain, in addition to an approximation to the true
Gaussian initial conditions (with a reduced magnitude).

Figure~\ref{fig:optimized}(c) shows that in the linear SWE case most
of the error in the reconstructed initial condition is due to
inaccuracies at wavenumbers $k=n \pi/\Delta x$, $n=1,2,3,4$
corresponding to the spacing of the observation points.  Comparing the
Fourier spectra shown in Figures~\ref{fig:optimized} (c) and~(d) indicates that results obtained with the linear and
nonlinear models are much closer for $N_{\text{obs}}=2$
observation points than for four observation points. For
$N_{\text{obs}}=4$ errors are smaller in the nonlinear SWE case, but
occur at the same wavenumbers as in the linear SWE case.  This
observation is confirmed in Figure~\ref{fig:err_vs_n} in
  which lower errors are also evident for the nonlinear case at all
spacings for four observation points.  Finally,
Figures~\ref{fig:optimized}(e) and (f) show similar results for the
larger spacing $\Delta x = 0.48$.  The errors in the Fourier
spectrum also occur at wavenumbers $k=n \pi/\Delta x$, confirming our
analysis.

\begin{figure}
\centering
\vspace*{-1.75cm}
\subfigure[$\Delta x = 0.09$, $N_{\text{obs}}=4$.]
{\includegraphics[width=0.44\textwidth]{\figpath/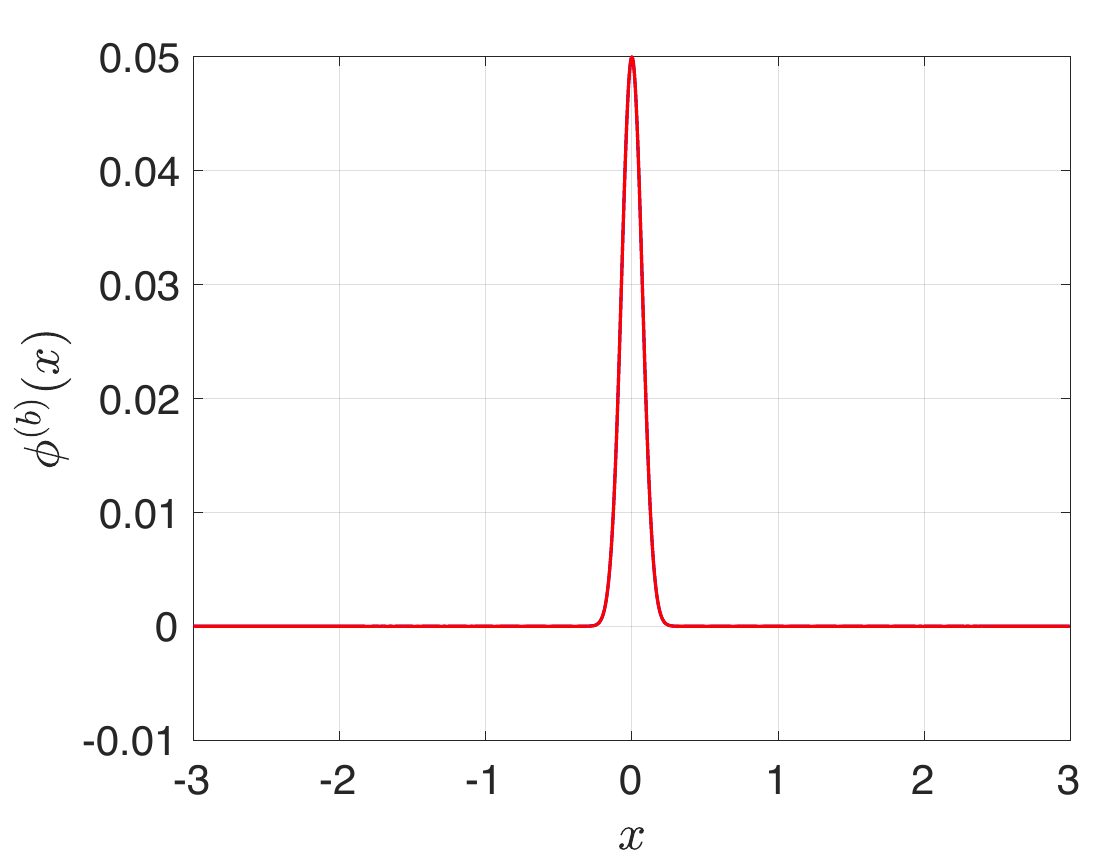}}\quad
\subfigure[$\Delta x = 0.375$, $N_{\text{obs}}=4$.]
{\includegraphics[width=0.44\textwidth]{\figpath/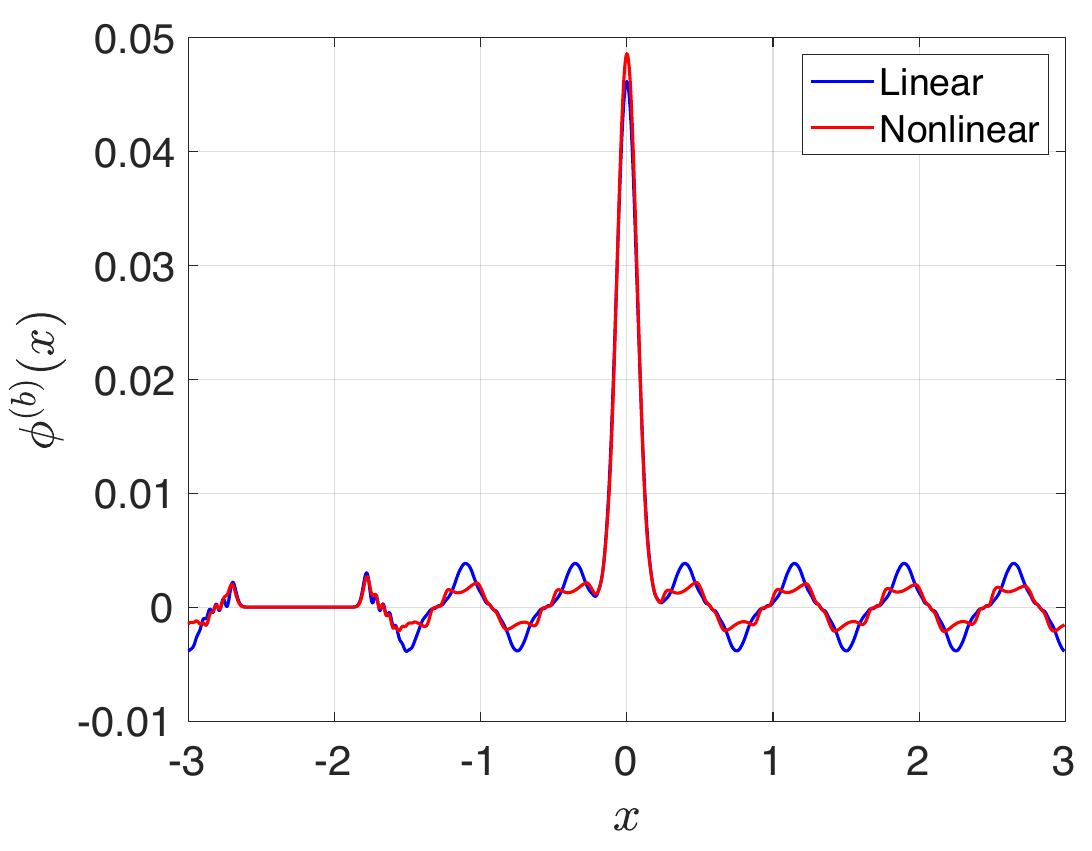}}
\subfigure[$\Delta x = 0.375$, $N_{\text{obs}}=4$.]
{\includegraphics[width=0.44\textwidth]{\figpath/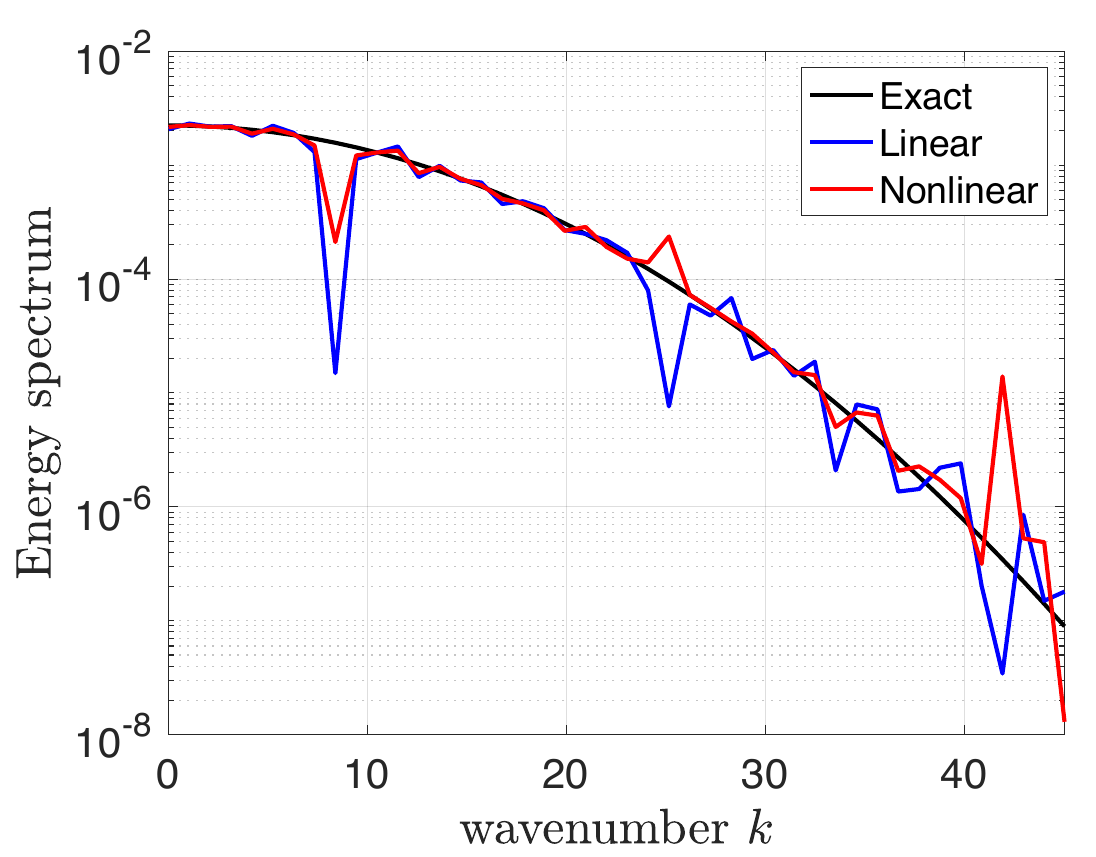}}\quad
\subfigure[$\Delta x = 0.375$, $N_{\text{obs}}=2$.]
{\includegraphics[width=0.44\textwidth]{\figpath/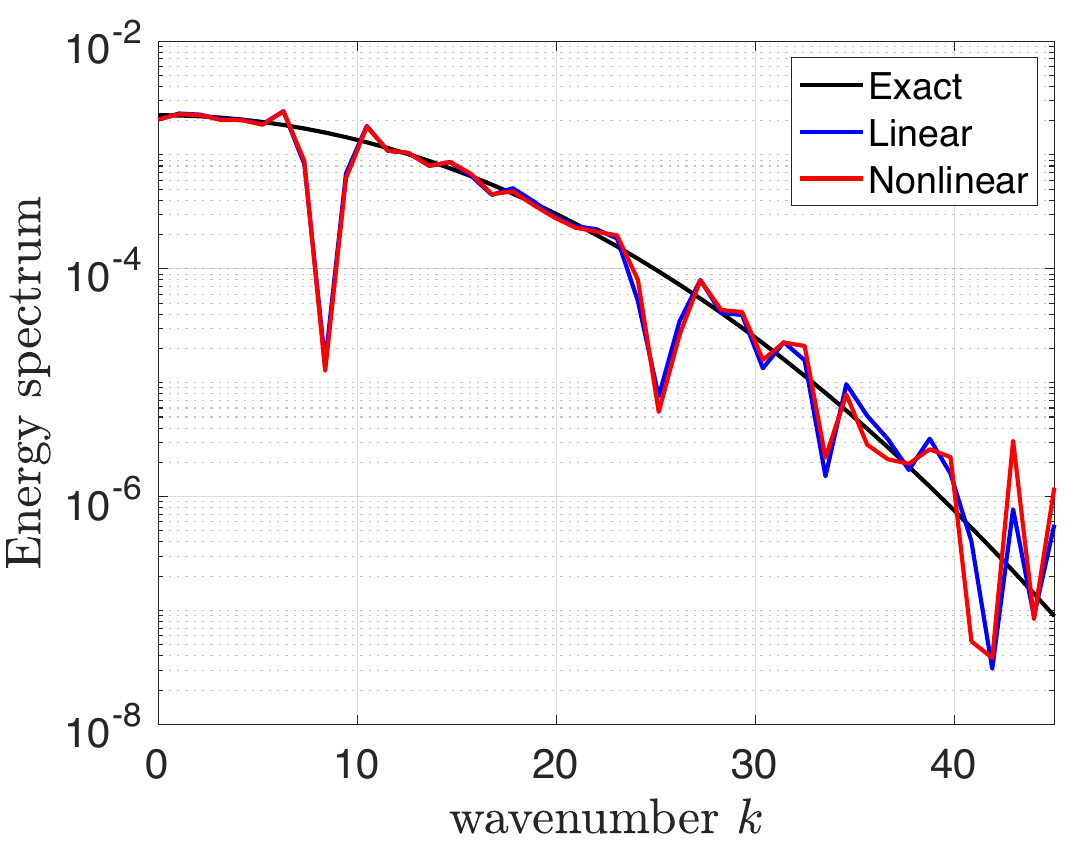}}
\subfigure[$\Delta x = 0.48$, $N_{\text{obs}}=2$.]
{\includegraphics[width=0.44\textwidth]{\figpath/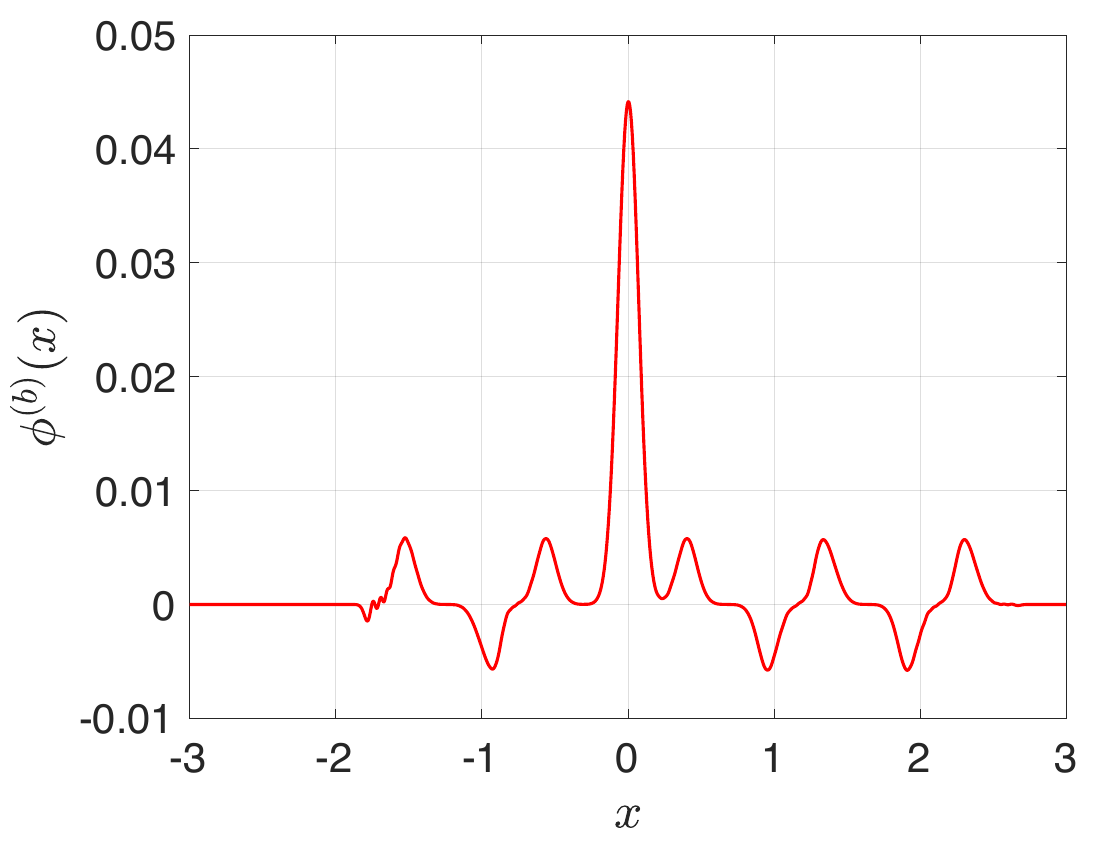}}\quad
\subfigure[$\Delta x = 0.48$, $N_{\text{obs}}=2$.]
{\includegraphics[width=0.44\textwidth]{\figpath/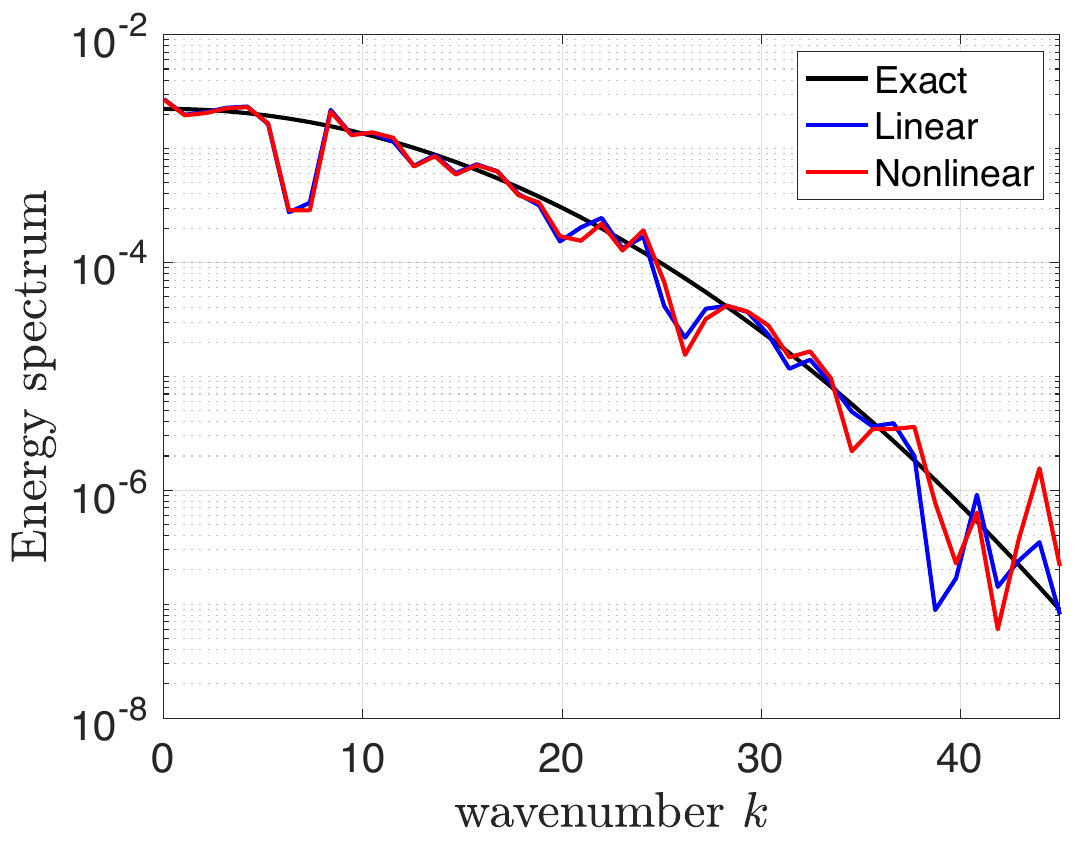}}
\caption{ \label{fig:optimized} {Comparison of} reconstructed
  initial conditions $\phi^{(b)}$ in physical and Fourier space
  {for cases satisfying and not satisfying the sufficient
    condition ~\eqref{eq:fp_condition} for convergence to the true
    initial data $\phi^{(t)}$}.
    (a)~Spacing $\Delta x = 0.09$ satisfying the
  sufficient condition, $N_{\text{obs}}=4$ {The initial
    conditions {$\phi^{(b)}$} reconstructed based on the
    linear and nonlinear models are both}
  indistinguishable from the exact initial condition {$\phi^{(t)}$}.  
  (b)~Spacing $\Delta x=0.375$ not satisfying the condition for convergence,
  $N_{\text{obs}}=4$.  {The reconstructions based on the} {linear and nonlinear equations converge to different {\it incorrect\/} initial conditions.}
  (c)~Fourier spectra of (b) compared with the spectrum {$\widehat{\phi}^{(t)}$} of the true initial condition. 
  (d)~As for (c), but with $N_{\text{obs}}=2$. 
  (e)~Spacing $\Delta x=0.48$ not satisfying the condition for convergence, $N_{\text{obs}}=2$. 
  {The reconstructions based on the} {linear and nonlinear equations both converge to the same {\it incorrect\/} initial conditions.}
  (f)~Fourier spectra of (e) compared with the spectrum
  {$\widehat{\phi}^{(t)}$} of the true initial condition.
  {The errors in the cases when convergence occurs to incorrect
    initial data are highly localized in Fourier space.}}
\end{figure}

Finally, we examine how the error in the reconstructed initial
condition $\phi^{(b)}$ changes as the uniform spacing $\Delta x$
between observation points varies between 0.0234 ($4h$, where $h$ is
the computational grid spacing) and 0.5 for two and four observation
points.  Figure~\ref{fig:transition} confirms that the numerical
results for the linear and nonlinear SWE data assimilations are
qualitatively consistent in three ways with the analysis in
Section~\ref{sec:fixedpoint}, in particular, with the behavior of the
function $|\widehat{\psi}(k)|$ shown in Figure~\ref{fig:psihat}.

First, when $\Delta x > 0.2$, the error generally increases
with increasing $\Delta x$.  Secondly, the error also increases
as $\Delta x\rightarrow 0$.  Thirdly, the error is ``spiky'' for
spacings $\Delta x>0.1$ due to the existence of a discrete set of
wavenumbers where $|\widehat{\psi}(k)| = 1$.  The spikes are not
visible in the case of two observation points since the peaks in
$|\widehat{\psi}(k)|$ are very broad in this case (cf.\
Figures~\ref{fig:psihat} (a) and (c)).  The error is significantly
lower at all spacings for the nonlinear SWE with $N_{\text{obs}}=4$,
while the linear and nonlinear SWE results are very similar for two
observation points.
\begin{figure}
\begin{center}
\includegraphics[width=0.7\textwidth]{\figpath/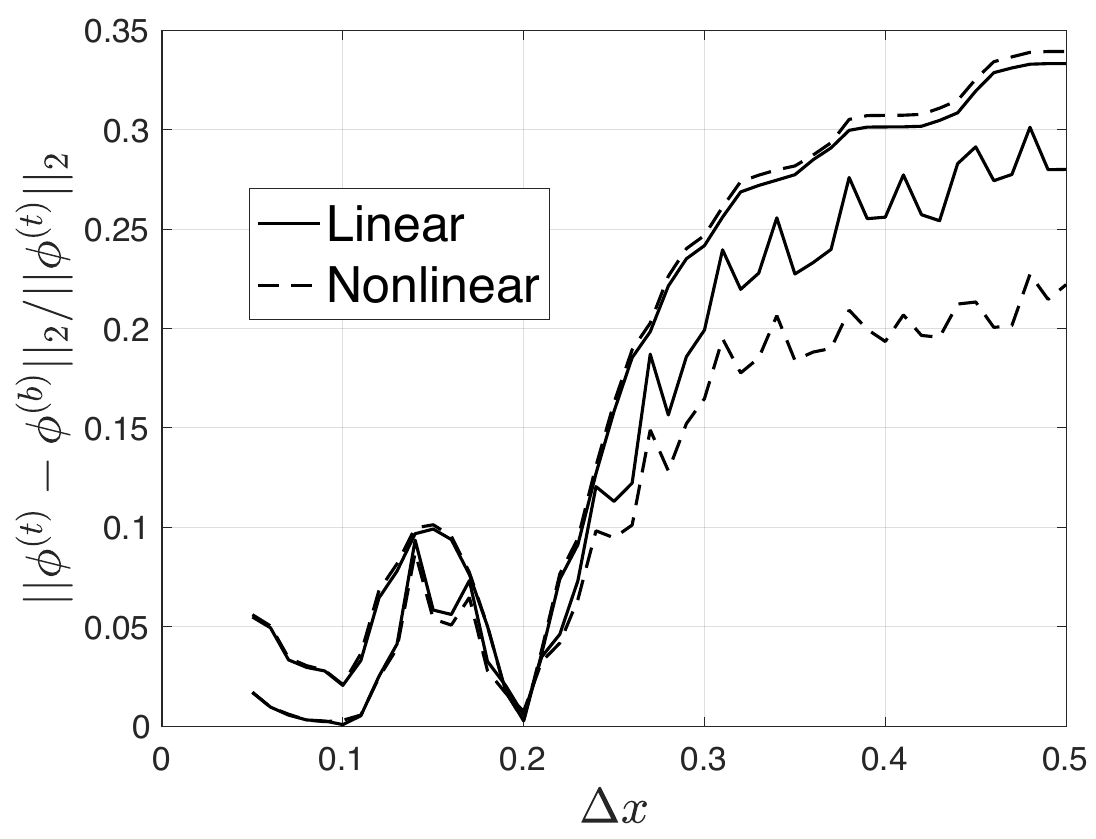}
\end{center}
\caption{Error in the reconstructed initial condition $\phi^{(b)}$ as
  a function of spacing between observation points $\Delta x$ for
  $N_{\text{obs}}=2$ (upper curves) and $N_{\text{obs}}=4$ (lower
  curves) after 1000 iterations. Note that the oscillations in the
  reconstruction error are similar to the oscillations in {the
    function} $|\widehat{\psi}(k)|$ {determining the convergence rate
    of the assimilation iterations} shown in
  Figure~\ref{fig:psihat}\label{fig:transition}.}
\label{fig:err_vs_n}
\end{figure}

\section{Conclusions\label{sec:concl}}
This paper considers the problem of estimating the initial conditions
for the one-dimensional SWE from wave height observations.  We
consider initial conditions with compact support and observations only
to the right side of this support (so only the right-going wave
is observed). This models the case of reconstructing the initial
conditions for a tsunami wave where observations typically do not
capture the total energy of the wave (which is the strict condition
for observability with one observation point~\cite{Zuazua:2005}).  It
is not clear {\it a priori\/} that the wave is observable under these
conditions, even with multiple observation points.

This data assimilation problem is framed as PDE-constrained
optimization in which a least-squares error functional is minimized
with respect to the initial condition. Optimal reconstructions are
computed with a gradient-based approach in which the cost functional
gradients are evaluated based on solutions of adjoint equations.
While in operational practice a ``discretize-then-differentiate''
approach is often adopted, in the present study we followed the
``optimize-then-discretize'' paradigm as it allowed us to derive a
sufficient condition for the convergence of iterations to the true
initial condition in the linear setting. Derivation and analysis of
this condition, which is easy to verify, are the key contributions of
this study.

The linear assimilation equations can be solved analytically, giving
the exact expression for the gradient of the cost functional.
Theorem~\ref{thm:conv} provides a sufficient condition for the
existence of a unique fixed point corresponding to convergence to the
true initial conditions.  It turns out that the algorithm cannot
converge at all for the case of a single observation point: at
least two observation points must be used.  However, the algorithm may
also fail to converge to the true initial conditions even with
multiple observation points.  A sufficient condition for the existence
of a unique fixed point is that at least one pair of observation
points is closer than $\lambda_{\text{min}}/2$, where $\lambda_{\text{min}}$ is the
minimum length scale of the true initial conditions. Note that
$\lambda_{\text{min}}/2$ is the Nyquist frequency associated to the initial
conditions. {These results also apply to the linear wave equation.}

{Furthermore, the analytical results suggest that in the discrete
  case the rate of convergence increases with increasing numbers of
  observation points.  This is confirmed by our numerical results.
  Interestingly, when the real line is discretized (rather than
  represented continuously) there are many observation point spacings
  larger than the critical spacing $\Delta x = \pi / k_{\text{max}}$
  that may still converge to the correct initial data since
    the ``bad'' spacings are represented only by a discrete set of
    values (this set, however, becomes increasingly dense as the
    resolution is refined, i.e., as $\Delta k \rightarrow 0$). It is
  therefore possible, especially with low-resolution
    discretizations, that with a ``lucky'' choice of $\Delta x > \pi /
    k_{\text{max}}$ one may still recover the initial conditions with
  widely spaced observation points.}

{We have assumed that the observations are continuous in time.
  This is a realistic assumption for tsunami monitoring, and is
  natural for the continuous infinite-dimensional PDE-based
    data assimilation problem that we consider.  In the case of
  discrete observations in time and wave speed $c$ it should be
  sufficient to sample the sea surface height at a frequency $f >
  c/(\lambda_{\text{min}})$ in order to capture the smallest scales of
  the wave.}

Finally, we verify the performance of the data assimilation algorithm
and the mathematical analysis for the case of a Gaussian initial
condition with zero initial guess.  The algorithm converges to the
true initial condition for both the linear and nonlinear SWE provided
there are at least two observation points and the points are
sufficiently close together (as determined by the theorem).  The rate
of convergence increases with the number of observation points, and at
least three points are needed for practically useful results.
Relative $L^2$ errors in the reconstructed initial conditions
of $O(10^{-2})$ are achieved in about 100 iterations and errors of
$O(10^{-4})$ are achieved in 500--1000 iterations for six observation
points.  {These results are essentially identical for the case of
an isolated bathymetry feature,  $\beta(x) = \exp(-10(x-3/2)^2)/10$ where the total mean depth is $1-\beta(x)$.}

We also confirmed that if the observation points are too widely spaced
(so the sufficient condition is not satisfied) the algorithm converges
to the wrong initial condition, for both linear and nonlinear SWE.
The error is due to underestimating the energy at discrete wavenumbers
$k=n\pi/\Delta x$, $n=1,2,3,\ldots$, where $\Delta x$ is the spacing
between the observation points. Interestingly, this failure to
converge to the true initial condition cannot be deduced from the
behaviour of the cost functional which decreases to small
values.

In addition to the simple gradient descent algorithm, we also solved
the minimization problem \eqref{eq:minJ} using two commonly used
nonlinear conjugate gradient methods: Fletcher--Reeves and
Polak-Ribi\'ere.  The results were consistent with the gradient
descent method although, as expected, convergence was usually faster.
Importantly, the conjugate gradient method did not modify the
conditions for convergence to the true initial conditions, i.e.
\eqref{eq:fp_condition} (see Figure~\ref{fig:Jerr_bad}).

Although this one-dimensional case is not physically realistic, it
allows for careful investigation of the potential and limitations of
variational data assimilation for the tsunami problem where only a
relatively small number of sparsely distributed observations are
possible. More specifically, it makes it possible to identify
sufficient conditions for convergence to the true initial condition
that should provide guidance for the two-dimensional problem.  In
particular, the Nyquist-like condition on the spacing of observation
points should carry over to the 2D case.  Importantly, the numerical
results have confirmed that the nonlinear problem behaves very
similarly to the linear problem we solved exactly.

As a next step, we plan to investigate to what extent the results reported 
here generalize to the SWE on 2D planar domains, or to a global ocean model 
\cite{Kevlahan/Dubos/Aechtner:2015}.  Mathematical analysis in such
settings is more complicated due to the geometry of the problem, although probably possible for the
linearized equations. A related question concerns using data 
 assimilation to infer bathymetry information from sea surface
height observations.  Cobelli et al.~\cite{Cobelli/etal:2017} recently investigated the related problem of determining the shape of a localized movement of the seafloor from surface observations. Bath\-y\-me\-try, together with the initial conditions, are the two major factors affecting the accuracy of
tsunami forecasting and available bathymetry data is often
incomplete or of insufficient resolution.

Finally, we also intend to investigate optimal placement of
observation points in the two-dimensional case including bathymetry
effects. As an important enabler for these efforts, our
dynamically adaptive wavelet method for the two-dimensional
SWE {on the sphere}~\cite{Kevlahan/Dubos/Aechtner:2015} opens up the possibility of
adapting the computational grid to improve the accuracy of the
assimilation.

\section*{Acknowledgements}
We are grateful to our colleague S.~Alama for his suggestion to use
the Fourier transform to analyze the fixed point of the gradient
descent algorithm.


\end{document}